\documentclass[11pt, autoref]{article}
\usepackage{graphicx}
\usepackage{relate}
\usepackage{times}
\usepackage{wrapfig}
\usepackage{pgfplots}
\usepackage{pgfplotstable}
\pgfplotsset{compat=1.8}
\usepackage[aboveskip=8pt,belowskip=-4pt]{caption}
\usepackage{subcaption}
\usepackage{xspace} 
\usepackage{enumitem}
\usepackage{cite, mathtools}
\usepackage[margin=1in]{geometry}
\usepackage{textcomp}
\usepackage{balance}  % for  \balance command ON LAST PAGE  (only there!)
\usepackage{xcolor}
\usepackage{amsmath}
\usepackage{algorithm}
\usepackage{algorithmicx}
\usepackage{cite}
\usepackage{float}
\usepackage{siunitx}
\usepackage{graphicx}
\usepackage{xspace}
\usepackage{thm-restate}
\usepackage{balance}  % for  \balance command ON LAST PAGE  (only there!)
\usepackage{booktabs} % For formal tables

\usepackage{marginnote} % MAB: alternative to marginpar. Is there something better?
\usepackage[hyphens]{url}
\usepackage{amssymb}
\usepackage{enumitem}
\graphicspath{{./fig/}}
\usepackage{amsthm}
\usepackage{mathtools}
\usepackage[noend]{algpseudocode}
\newtheorem{theorem}{Theorem}
\newtheorem{lemma}[theorem]{Lemma}
\newtheorem{corollary}[theorem]{Corollary}

\newtheorem{observation}{Observation}
\usepackage{comment}
\usepackage{color}

\usepackage{tikz}
\usetikzlibrary{fadings}
\usepackage{hyperref}
\usepackage{cleveref}

\usepackage{titling}
\thanksmarkseries{arabic}

\date{}

\begin{document}
%!TEX root =  main.tex

\title{Cutting Polygons into Small Pieces with Chords:\\
	 Laser-Based Localization}

\author{
Esther M. Arkin\thanks{Stony Brook University, NY, USA. Email:
\texttt{\{esther.arkin,rathish.das,joseph.mitchell\}@stonybrook.edu}.}
\and
Rathish Das\protect\footnotemark[1]
\and
Jie Gao\thanks{Rutgers University, NJ, USA. Email:\texttt{jg1555@cs.rutgers.edu}.}
\and
Mayank Goswami\thanks{Queens College of CUNY, NY, USA. Email:\texttt{mayank.goswami@qc.cuny.edu}.}
\and
Joseph S. B. Mitchell\protect\footnotemark[1]
\and
Valentin Polishchuk\thanks{Link\"oping University, Norrk\"oping, Sweden. Email:
	\texttt{valentin.polishchuk@liu.se}.} 
\and
Csaba D. T\'oth\thanks{California State University Northridge, CA;
	and Tufts University, MA, USA. Email:
	\texttt{csaba.toth@csun.edu}.}
}

%%% Local Variables:
%%% mode: latex
%%% TeX-master: "main.tex"
%%% End:

\maketitle
\thispagestyle{empty}
\sloppy

\begin{abstract}
	Motivated by indoor localization by tripwire lasers, we study the problem of cutting a polygon into small-size pieces, using the chords of the polygon. Several versions are considered, depending on the definition of the ``size'' of a piece. In particular, we consider the area, the diameter, and the radius of the largest inscribed circle as a measure of the size of a piece. We also consider different objectives, either minimizing the maximum size of a piece for a given number of chords, or minimizing the number of chords that achieve a given size threshold for the pieces. We give hardness results for polygons with holes and approximation algorithms for multiple variants of the problem.
\end{abstract}

\clearpage
\pagenumbering{arabic} 
%\input{vision}
%\input{intro}
%!TEX root =  main.tex

% MAB: I'm not sure whether we'll say IO or I/O. So I'm deferring the decision.

% Footnote commands.

\newcommand\e\emph\newcommand\eps{\ensuremath\varepsilon\xspace}
\renewcommand\d{\ensuremath{\delta}\xspace}\renewcommand\k{\ensuremath{k}\xspace}\renewcommand\P{\ensuremath{P}\xspace}\renewcommand\g{\ensuremath{\gamma}\xspace}
\newcommand\bd{\ensuremath{\partial}\xspace}
\renewcommand\mm{\textsf{MinMeasure}\xspace}\newcommand\mlm{\textsf{Min-LaserMeasure}\xspace}
\newcommand\ma{\textsf{MinArea}\xspace}\newcommand\mla{\textsf{Min-LaserArea}\xspace}
\newcommand\md{\textsf{MinDiameter}\xspace}\newcommand\mld{\textsf{Min-LaserDiameter}\xspace}
\newcommand\mc{\textsf{MinCircle}\xspace}\newcommand\mlc{\textsf{Min-LaserCircle}\xspace}
\newtheorem{proposition}[theorem]{Proposition}
\newenvironment{reflemma}[1]{\medskip\parindent 0pt{\bf Lemma \ref{#1}.}\em }{\vspace{1em}}
\newenvironment{refthm}[1]{\medskip\parindent 0pt{\bf Theorem \ref{#1}.}\em }{\vspace{1em}}

\newcommand{\footnotenonumber}[1]{{\def\thempfn{}\footnotetext{#1}}}
\newcommand{\footnotetight}[1]{\footnote{\renewcommand\baselinestretch{1}\footnotesize #1}}

\newcommand{\nmax}{N_{\text{\rm max}}}
\newcommand{\io}{I/O\xspace}
\newcommand{\ios}{I/Os\xspace}

\newcommand{\punt}[1]{}
\newcommand{\calU}{{\cal U}}
\newcommand{\calC}{{\cal C}}
\newcommand{\calB}{{\cal B}}
\newcommand{\calS}{{\cal S}}
\newcommand{\calF}{{\cal F}}
\newcommand{\calD}{{\cal D}}

\newcommand{\FP}{\textsc{False Positive}\xspace}
\newcommand{\FN}{\textsc{False Negative}\xspace}
\newcommand{\FPs}{\textsc{False Positives}\xspace}
\newcommand{\FNs}{\textsc{False Negatives}\xspace}
\newcommand{\OL}{\textsc{Online}\xspace}
\newcommand{\Scalable}{\textsc{Scalable}\xspace}
\newcommand{\odp}{\textsc{odp}\xspace}
\newcommand{\oedp}{\textsc{oedp}\xspace}
\newcommand{\tedp}{\textsc{tedp}\xspace}
\newcommand{\oedpfp}{\textsc{oedpfp}\xspace}

\newcommand{\oaedp}{\textsc{oaedp}\xspace}

\newcommand{\taedp}{\textsc{taedp}\xspace}

\makeatletter
%\renewtheoremstyle{plain}% Adds automatic line break, if heading is too long
%  {\item{\theorem@headerfont ##1\ ##2\theorem@separator}{\bf.}~}
%  {\item{\theorem@headerfont ##1\ ##2\ (##3)\theorem@separator}{\bf.}~}
%\makeatother
%
%{\theoremheaderfont{\upshape\bfseries}
% \theorembodyfont{\normalfont\em}
%\newtheorem{definition}{Definition}}

\makeatletter
\def\@copyrightspace{\relax}
\makeatother

\newcommand{\defn}[1]       {{\textit{\textbf{\boldmath #1}}}}
\newcommand{\pparagraph}[1]{\vspace{0.09in}\noindent{\bf \boldmath #1.}}
\renewcommand{\paragraph}[1]{\vspace{0.09in}\noindent{\bf \boldmath #1}}
\newcommand{\poly}{\mbox{poly}}
\newcommand{\polylog}{\mbox{polylog}}
%%%% VARIABLE NAMES %%%%

\newcommand{\Ns}{N}
\newcommand{\Nk}{u}
%% MAB: do we want stream to be $S$,  $\calS$, or something else?
\newcommand{\stream}{S}

\date{}

\newcommand{\namedcomment}[3]{{\color{#2} #1: #3}}

% Uncomment to turn off comments
%\renewcommand{\namedcomment}[3]{}

\newcommand{\MinLaserConvexArea}{{\sc MinLaser-ConvexArea}\xspace}
\newcommand{\Att}[1]{\namedcomment{Att}{red}{#1}}

\newcommand{\estieESA}[1]{\namedcomment{Estie}{red}{#1}}
\newcommand{\csabaESA}[1]{\namedcomment{Csaba}{blue}{#1}}
\newcommand{\rathishESA}[1]{\namedcomment{Rathish}{purple}{#1}}
\newcommand{\joeESA}[1]{\namedcomment{Joe}{olive}{#1}}
\newcommand{\jieESA}[1]{\namedcomment{Jie}{magenta}{#1}}
\newcommand{\valESA}[1]{\namedcomment{Val}{cyan}{#1}}
\newcommand{\mayankESA}[1]{\namedcomment{Mayank}{blue}{#1}}

\newcommand{\estie}[1]{\namedcomment{Estie}{red}{#1}}
\newcommand{\csaba}[1]{\namedcomment{Csaba}{blue}{#1}}
\newcommand{\rathish}[1]{\namedcomment{Rathish}{purple}{#1}}
\newcommand{\joe}[1]{\namedcomment{Joe}{olive}{#1}}
\newcommand{\jie}[1]{\namedcomment{Jie}{magenta}{#1}}
\newcommand{\val}[1]{\namedcomment{Val}{cyan}{#1}}
\newcommand{\mayank}[1]{\namedcomment{Mayank}{blue}{#1}}

\renewcommand{\rathish}[1]{} 
\renewcommand{\estie}[1]{} 
\renewcommand{\csaba}[1]{} 
\renewcommand{\joe}[1]{} 
\renewcommand{\jie}[1]{} 
\renewcommand{\val}[1]{} 
\renewcommand{\mayank}[1]{} 

\newcommand{\rathishI}[1]{} 
\newcommand{\estieI}[1]{} 
\newcommand{\csabaI}[1]{} 
\newcommand{\joeI}[1]{}
\newcommand{\jieI}[1]{} 
\newcommand{\valI}[1]{} 
\newcommand{\mayankI}[1]{}

% \begin{animateinline}[autoplay,loop]{2}%
%%   \randomcolor{randomcolora}
%%   \randomcolor{randomcolorb}
%%   \randomcolor{randomcolorc}

%%   \noindent\fadingtext{left color=randomcolora,right color=randomcolorb,middle color=randomcolorc!80!black}
%%              {\sf \scriptsize \sloppy \parbox{6.5in}{Rob:  #1}}
%%              %\newframe
%%              %\end{animateinline}
%% }
\newcommand{\varK}{24\xspace}
\renewcommand{\epsilon}{\varepsilon}
\newcommand{\bet}{B$^{\varepsilon}$-tree\xspace}
\newcommand{\bets}{B$^{\varepsilon}$-trees\xspace}
\newcommand{\bonehalftree}{B$^{1/2}$-tree\xspace}
\newcommand{\bonehalftrees}{B$^{1/2}$-trees\xspace}
\newcommand{\pf}{popcorn filter\xspace}

%% References
\newcommand{\appref}[1]         {Appendix~\ref{app:#1}}
\newcommand{\applabel}[1]    {\label{app:#1}}

\newcommand{\chapref}[1]        {Chapter~\ref{chap:#1}}
\newcommand{\secref}[1]         {Section~\ref{sec:#1}}
\newcommand{\seclabel}[1]    {\label{sec:#1}}
\newcommand{\subsecref}[1]      {Subsection~\ref{subsec:#1}}
\newcommand{\subseclabel}[1]    {\label{subsec:#1}}
\newcommand{\secreftwo}[2]      {Sections \ref{sec:#1} and~\ref{sec:#2}}
\newcommand{\secrefthree}[3]    {Sections \ref{sec:#1}, \ref{sec:#2}, and \ref{sec:#3}}
\newcommand{\secreffour}[4]     {Sections \ref{sec:#1}, \ref{sec:#2}, \ref{sec:#3}, and~\ref{sec:#4}}
\newcommand{\quanlabel}[1] {\label{quan:#1}}
\newcommand{\quanref}[1]  {Quantity~\ref{quan:#1}}
\newcommand{\figlabel}[1]   {\label{fig:#1}}
\newcommand{\figref}[1]         {Figure~\ref{fig:#1}}
\newcommand{\figreftwo}[2]      {Figures \ref{fig:#1} and~\ref{fig:#2}}
\newcommand{\tabref}[1]         {Table~\ref{tab:#1}}
\newcommand{\tablabel}[1]   {\label{tab:#1}}
\newcommand{\stref}[1]          {Step~\ref{st:#1}}
\newcommand{\thmlabel}[1]   {\label{thm:#1}}
\newcommand{\thmref}[1]         {Theorem~\ref{thm:#1}}
\newcommand{\thmabbrevref}[1]         {Thm.~\ref{thm:#1}}
\newcommand{\claimlabel}[1]         {\label{claim:#1}}
\newcommand{\claimref}[1]         {Claim~\ref{claim:#1}}
\newcommand{\thmreftwo}[2]      {Theorems \ref{thm:#1} and~\ref{thm:#2}}
\newcommand{\lemlabel}[1]   {\label{lem:#1}}
\newcommand{\lemref}[1]         {Lemma~\ref{lem:#1}}
\newcommand{\algolabel}[1]   {\label{alg:#1}}
\newcommand{\algoref}[1]         {Algorithm~\ref{alg:#1}}
\newcommand{\lemreftwo}[2]      {Lemmas \ref{lem:#1} and~\ref{lem:#2}}
\newcommand{\lemrefthree}[3]    {Lemmas \ref{lem:#1}, \ref{lem:#2}, and~\ref{lem:#3}}
\newcommand{\corlabel}[1]   {\label{cor:#1}}
\newcommand{\corref}[1]         {Corollary~\ref{cor:#1}}
\newcommand{\nonlabel}[1]    {\label{blank:#1}}
\newcommand{\nonref}[1]          {~(\ref{blank:#1})}
\newcommand{\eqlabel}[1]    {\label{eq:#1}}
\newcommand{\eqreff}[1]          {(\ref{eq:#1})}
\renewcommand{\eqref}[1]          {Eq.~\ref{eq:#1}}
\newcommand{\eqreftwo}[2]       {(\ref{eq:#1}) and~(\ref{eq:#2})}
\newcommand{\ineqlabel}[1]    {\label{ineq:#1}}
\newcommand{\ineqref}[1]        {Inequality~(\ref{ineq:#1})}
\newcommand{\ineqreftwo}[2]     {Inequalities (\ref{ineq:#1}) and~(\ref{ineq:#2})}
\newcommand{\invref}[1]         {Invariant~\ref{inv:#1}}
\newcommand{\deflabel}[1]    {\label{def:#1}}
\newcommand{\defref}[1]         {Definition~\ref{def:#1}}
\newcommand{\propref}[1]        {Property~\ref{prop:#1}}
\newcommand{\propreftwo}[2]     {Properties~\ref{prop:#1} and~\ref{prop:#2}}
\newcommand{\proplabel}[1]        {\label{prop:#1}}

\newcommand{\caseref}[1]        {Case~\ref{case:#1}}
\newcommand{\casereftwo}[2]     {Cases \ref{case:#1} and~\ref{case:#2}}
\newcommand{\lilabel}[1]        {\label{li:#1}}
\newcommand{\liref}[1]          {line~\ref{li:#1}}
\newcommand{\Liref}[1]          {Line~\ref{li:#1}}
\newcommand{\lirefs}[2]         {lines \ref{li:#1}--\ref{li:#2}}
\newcommand{\Lirefs}[2]         {Lines \ref{li:#1}--\ref{li:#2}}
\newcommand{\lireftwo}[2]       {lines \ref{li:#1} and~\ref{li:#2}}
\newcommand{\lirefthree}[3]     {lines \ref{li:#1}, \ref{li:#2}, and~\ref{li:#3}}
\newcommand{\exref}[1]          {Exercise~\ref{ex:#1}}
\newcommand{\princref}[1]       {Principle~\ref{prop:#1}}

\newcommand{\obslabel}[1]   {\label{obs:#1}}

\newcommand{\resultref}[1]         {Result~\ref{result:#1}}
\newcommand{\resultlabel}[1]   {\label{result:#1}}
\newcommand{\resultreftwo}[2]      {Results~\ref{result:#1} and~\ref{result:#2}}
\newcommand{\resultrefthree}[3]    {Results~\ref{result:#1}, \ref{result:#2}, and~\ref{result:#3}}
\newcommand{\resultrefthrough}[2]      {Results~\ref{result:#1}-\ref{result:#2}}

\newcommand{\area}{{\rm area}}
\newcommand{\conv}{{\rm conv}}
\newcommand{\diam}{{\rm diam}}
\renewcommand{\per}{{\rm per}}
\newcommand{\len}{{\rm len}}
\newcommand{\opt}{{\rm OPT}\xspace}
\newcommand{\alg}{{\rm ALG}}

\newcommand{\denselist}{\itemsep 0pt\parsep=1pt\partopsep 0pt}
\newcommand{\bitem}{\begin{itemize}\denselist}
	\newcommand{\eitem}{\end{itemize}}
\newcommand{\benum}{\begin{enumerate}\denselist}
	\newcommand{\eenum}{\end{enumerate}}

\DeclarePairedDelimiter\ceil{\lceil}{\rceil}
\DeclarePairedDelimiter\floor{\lfloor}{\rfloor}

%%% Local Variables:
%%% mode: latex
%%% TeX-master: "main.tex"
%%% End:

\section{Introduction}

Indoor localization is a challenging and important problem. While GPS technology is very effective outdoors, it generally performs poorly inside buildings, since GPS depends on line-of-sight to satellites. Thus, other techniques are being considered for indoor settings. One of the options being investigated for localization and tracking is to use one-dimensional tripwire sensors~\cite{He2004-xs} such as laser beams, video cameras with a narrow field of view~\cite{Zahnd_undated-pw}, and pyroelectric or infrared sensors~\cite{Gopinathan2003-oo,Gustafson1982-ng}. In these approaches, multiple sensors emitting directional signal beams
are deployed in an environment, with the beams inducing an arrangement that cuts the domain into cells, allowing one to track the movement of a mobile target from one cell to another when it crosses the signal beam. Since the accuracy of the localization depends on the sizes of the cells, it is desirable to cut the polygon into \e{small} pieces. With such beam deployment, one can also ensure that no ``large'' object can be ``hidden'' in the domain, since any such object will necessarily intersect one of the beams.

In the literature there have been studies of target localization and tracking using such ``tripwire'' sensors. Zheng, Brady, and Agarwal~\cite{Zheng2007-fl} consider general models of ``boundary sensors'' that are triggered when an object crosses them. They assume that the position of the sensors is already given and consider the signal processing problem of determining the location and trace of a target by the spatial and temporal sequence of the laser beams crossed by the target. In this paper, we focus on the problem of optimizing the placement of signal beam sensors to minimize the ambiguity of target location within each cell.

%\subsection*{Problem Formulation and Notation}
\paragraph{Problem Formulation and Notation.}
We study various versions of the laser cutting problem. The input polygon, denoted by~\P, is a closed polygonal domain (i.e., a connected compact set in $\mathbb{R}^2$ with piecewise linear boundary) having a total of $n$ vertices, $r$ of which are reflex (having internal angle greater than $\pi$).
% When we consider simple polygons, $r$ will denote the number of reflex vertices of \P.
The terms ``cut'' and ``laser'' will be used interchangeably to denote a chord of \P, i.e., a maximal line segment in \P whose relative interior lies in the interior of \P.
%\csaba{We do not allow a chord to pass through a reflex vertex of \P.}
The \e{measure} (or \e{size}) of a cell in the arrangement will be (a)~the cell's area, (b)~its diameter (defined as the maximum Euclidean distance between two points of the cell), or (c)~the radius of the largest inscribed disk within the cell.

% We consider two formulations of the optimization problem: (1)~require that the size of each cell in the partition is at most \d and seek to minimize the number of laser cuts to achieve the desired cell sizes, or (2)~fix the number, $k$, of laser cuts and seek to minimize the maximum cell size,~\d.

For each measure, we consider two formulations of the optimization problem:
\begin{itemize}
\item%\ma (resp., \md):
\mm: Given a positive integer $k$, determine how to place $k$ laser beams in \P to minimize the maximum measure, \d, of a cell in the arrangement of the lasers.
\item%\mla (resp., \mld):
\mlm: Given $\d>0$, determine the smallest number of laser beams to cut \P into cells each of measure at most \d.
\end{itemize}
%Here \emph{measure} is one of the measures defined above. 
In \mlm, no generality is lost by taking the cell size bound, \d, to be~1. We assume that the optimal solution is greater than a constant $c$; otherwise, the problem can be solved optimally in $O(n^{\textrm{poly}(c)})$ time (in the real RAM model of computation, standard for geometric algorithms) by reducing it to a mathematical program whose variables are the locations of the lasers endpoints on the boundary of \P (the space of the variables would be split into regions of fixed combinatorial types for all the lasers, and in each region, the measures for the cells of the partition of \P will be explicitly written and optimized---since each cell has $\textrm{poly}(c)=O(1)$ complexity, the optimization problem will be of constant size).
%val added, following socg rev3 comment
It may be interesting to investigate also the opposite scenario and obtain efficient algorithms for minimizing the measures using a small given number of lasers.
Further variants of the problem may be defined. One %may require, e.g., that the cells in the partition are convex. Since the partition is defined by arrangement of straight lines, the non-convexity of the pieces may come only from reflex vertices of \P, which may be handled by simply adding $r$ lasers along bisectors of the reflex vertices (since in this version $r$ is a lower bound on the minimum number of lasers required such placement adds only a constant to the approximation ratio of an approximation algorithm for \mlm). Another
possible requirement is to use only axis-aligned lasers---in fact, with this restriction 
(of primarily theoretical interest)
% arguably interesting from theory perspective only
we obtain better approximations than for the more general case of unrestricted-orientation lasers.

%\subsection*{Results}
\paragraph{Results.}
We give hardness results and approximation algorithms for several variants of the problems, using a variety of techniques. Specifically,\begin{itemize}
\item Section~\ref{sec:hardness} proves hardness of our problems in polygons \e{with holes}: we show that it is NP-hard to decide whether one can split the domain into pieces of measure at most \d, using a given number \k of lasers (this holds for any of the measures, which implies that both \mm and \mlm are hard for polygons with holes). Our hardness reductions hold using axis-parallel lasers, as well, which implies that the problem is hard with or without the restriction to axis-aligned lasers.
\item Section~\ref{sec:minlaser-area} gives an $O(\log r)$-approximation for \mla in \e{simple} polygons. The algorithm ``unrefines'' the ray shooting subdivision by Hershberger and Suri~\cite{suri}, merging the triangles bottom-up along the decomposition tree; the merging stops whenever the next merge would create a cell of area greater than \d, implying that the boundaries between the merged cells can be charged to disjoint parts of \P of area more than \d. The lasers are then put along the cell boundaries of the coarsened subdivision; since the subdivision is obtained by cutting out $O(1)$ children from parents in a tree on the original subdivision (where the children were separated from parents by polygonal chains of $O(1)$ complexity), we can charge these $O(1)$ lasers to the intersection of \opt with an area of more than \d. The remaining large pieces in the coarsened subdivision (e.g., triangles of area more than \d in the initial triangulation) are cut with a suitable grid of lasers, which is within a constant factor of optimal subdivision for each piece. The $O(\log n)$ approximation factor then follows from the fact that each laser could pass through $O(\log n)$ cells of the original subdivision (the subdivision's core property). To bring the approximation factor down to $O(\log r)$ we decompose \P into convex pieces with a decomposition whose stabbing number is $O(\log r)$ (a result, which may be of independent interest) and use the same scheme as with the Hershberger--Suri decomposition.
\item In Section~\ref{sec:minlaser-diameter} we present a bi-criteria approximation to the diameter version for \e{simple} polygons: if \k lasers can cut \P into pieces of diameter at most \d, we find a cutting with at most 2\k lasers into $O(\d)$-diameter pieces. In Section~\ref{sec:k-laser-mindiameter} we use the bi-criteria algorithm to give a constant-factor approximation to \md. Both algorithms use only axis-aligned lasers, yielding the same approximation guarantees for the versions with general-direction lasers and with axis-aligned lasers.
\item Section~\ref{sec:orthogonal-laser} gives a constant-factor approximation to \mld and \mla in \e{simple} polygons under the restriction that the lasers are axis-aligned. The algorithms are based on ``histogram decomposition'' with constant stabbing number and solving the problems in each histogram separately. 
\item In Section~\ref{sec:polygon-holes} we give a bi-criteria approximation to the diameter version in polygons \e{with holes} under the restriction that lasers are axis-parallel. The algorithm is similar to the one for simple polygons in that they both use a grid; however, everything else is different: in simple polygons we place lasers along grid lines, while in polygons with holes the grid lines just subdivide the problem (in fact, we consider the vertical and the horizontal strips separately). More importantly, even though we place axis-aligned lasers in both simple and nonsimple polygons, for the former we approximate cutting with arbitrary-direction lasers, while for the latter only cuttings with axis-aligned lasers (approximating cuttings with general-direction lasers in polygons with holes is open). We use the bi-criteria algorithm to give a constant-factor approximation to \md in polygons with holes---this part is the same as for simple polygons.
\item Section~\ref{sec:circle} gives an $O(\log\opt)$-approximation for \mlc in polygons \e{with holes}. The algorithm is based on a reduction to the SetCover problem.
\end{itemize}

Table~\ref{tab:results} summarizes our results.
The running times of our algorithms depend on the output complexity, which may depend on the size (area, perimeter, etc.) of \P. Some of our algorithms can be straightforwardly made to run in strongly-polynomial time, producing a strongly-polynomial-size representation of the output; for others, such conversion---which in general is outside our scope---is not easily seen.
Many versions of the problem still remain open. For simple polygons, despite considerable attempts, we have neither hardness results nor polynomial-time algorithms to compute an optimal solution; all of our positive results are approximation algorithms.

\begin{table}[!htb]
	\centering
	\begin{tabular}{ |c|c|c|c|c|  }
		\hline
		\multicolumn{1}{|c|}{ } & \multicolumn{2}{c|}{Axis-Parallel Lasers} & \multicolumn{2}{c|}{Unrestricted-Direction Lasers} \\
		\hline
		& \mlm & \mm & \mlm & \mm\\
		\hline
		Area  & $O(1)$ \ref{sec:orthogonal-laser}  & OPEN & $O(\log r)$ \ref{sec:minlaser-area} & OPEN \\
		\hline
		Diameter & $O(1)$ \ref{sec:orthogonal-laser}  & $O(1)$* \ref{sec:k-laser-mindiameter}, \ref{sec:polygon-holes} & bi-critreria \ref{sec:minlaser-diameter}  & $O(1)$ \ref{sec:k-laser-mindiameter} \\
		\hline
		In-circle radius & $O(\log \mbox{OPT})$* \ref{sec:circle} & OPEN & $O(\log \mbox{OPT})$*  \ref{sec:circle} & OPEN \\
		\hline
	\end{tabular}
\caption{Approximations for simple polygons. The results marked with asterisks apply also to polygons with holes (either directly or with a similar/extended algorithm).}\label{tab:results}
\end{table}

%\rathish{update the table once section 4 is finished.}
%\rathish{Please uncomment line \#81 in defines.tex to turn off all the colorful comments.}

%\paragraph{Related Previous Work.}
\noindent\textbf{Related Previous Work.}
Previous results on polygon decomposition~\cite{Keil00} use models that do not support laser cuts or are restricted to convex bodies. For example, \emph{Borsuk's conjecture}~\cite{Borsuk32,JenrichB14,KK93} seeks to partition a convex body of unit diameter in $\mathbb{R}^d$ into the minimum number of pieces of diameter less than one.
Conway's \e{fried potato problem}~\cite{croftunsolved,bezdek1995solution} seeks to minimize the maximum in-radius of a piece after a given number of \e{successive} cuts by hyperplanes for a convex input polyhedron in $\mathbb{R}^d$.  
Croft et al.~\cite[Problem~C1]{croftunsolved} raised a variant of the problem in which a convex body is partitioned by an \e{arrangement} of hyperplanes (i.e., our problem in $\mathbb{R}^d$), but no results have been presented.

Equipartition problems ask to partition convex polygons into convex pieces all having the same area or the same perimeter (or other measures)~\cite{barany2010equipartitioning,Blagojevic2014,Karasev2014,kostitsyna2015balanced,nandakumar2012fair,soberon2012}.
%aronov2010convex,karasev2010equipartition%
In these problems, the partition is not restricted to chords (or hyperplanes). Topological methods are used for existential results in this area, and very few algorithmic results are known~\cite{armaselu2015algorithms}.
Another related problem is the family of so-called \emph{cake cutting problems}~\cite{robertson1998cake,GH2005}, %{woeginger2007complexity}
in which an infinite straight line ``knife'' is used to cut a convex ``cake'' into (convex) pieces that represent a ``fair'' division into portions.
In contrast, we are interested in cutting \emph{nonconvex} polygons into connected pieces.

In \cite{bose1998polygon} several variants of Chazelle's result from \cite{chazelle1982theorem} were explored, including cutting the polygon along a chord to get approximately equal areas of the two resulting parts.
Yet another related problem is that of ``shattering'' with arrangements \cite{Freimer-et-al}, in which one seeks to isolate objects in cells of an arrangement of a small number of lines, but without consideration of the size of the cells (as is important in our problem). 
\section{Hardness in Polygons with Holes}\label{sec:hardness}
We show that for all three measures (area, diameter, the radius of the largest inscribed circle) it is NP-hard to decide whether a given polygon \P \emph{with holes} can be divided into pieces of small measure using a given number of lasers, both for unrestricted-orientation and axis-aligned lasers. However, it is currently open whether these problems remain NP-hard for simple polygons.%; and for the case where the resulting pieces are constrained to be convex.

We prove hardness by reduction from the 3SAT problem. Our polynomial-time reduction is similar to previous reductions for line cover problems, which are geometric variants of set cover~\cite{LangermanM05}. In particular, Megiddo and Tamir~\cite{MegiddoT82} proved that the \textsc{LineCover} problem is NP-complete: Given $n$ points in the plane and an integer $k$, decide whether the points can be covered by $k$ lines. Hassin and Megiddo~\cite{HassinM91} proved hardness for \textsc{MinimumHittingHorizontalUnitSegments} problem: Given $n$ horizontal line segments in the plane, each of unit length, and an integer $k$, decide whether there exists a set of $k$ axis-parallel lines that intersects all $n$ segments.
Our reduction is based on the idea of Hassin and Megiddo, but requires some adjustments to generate a subdivision of a polygon.

\begin{theorem}\label{thm:hardness1}
In a polygon with holes, both \ma and \mla are NP-hard (with or without the axis-aligned lasers restriction).
\end{theorem}
\begin{proof}
We reduce from 3-SAT. Let $\Phi$ be a boolean formula in 3CNF with $m$ clauses $c_1,\ldots ,c_m$, and $n$ variables $x_1,\ldots , x_n$. We construct an orthogonal polygon $P$ with holes and an integer $k$ such that $\Phi$ is satisfiable if and only if $P$ can be subdivided into regions of area at most $2$ using $k$ lasers. (The reduction goes through with or without the restriction that all lasers are axis-parallel).

We construct a polygon $P$ from the rectangle $B=[0,7m+2]\times [0,3n+4]$ by carving rectangular ``rooms'' connected by narrow corridors.
The rooms are pairwise disjoint and they each have area of $2$. The corridors are axis-parallel, run between opposite sides of the bounding box $B$, and their width is $1/(100\max\{m,n\})$. See Fig.~\ref{fig:reduction} for an illustration.

 \begin{figure*}[h]
 \centering
 	\includegraphics[width=0.75\textwidth]{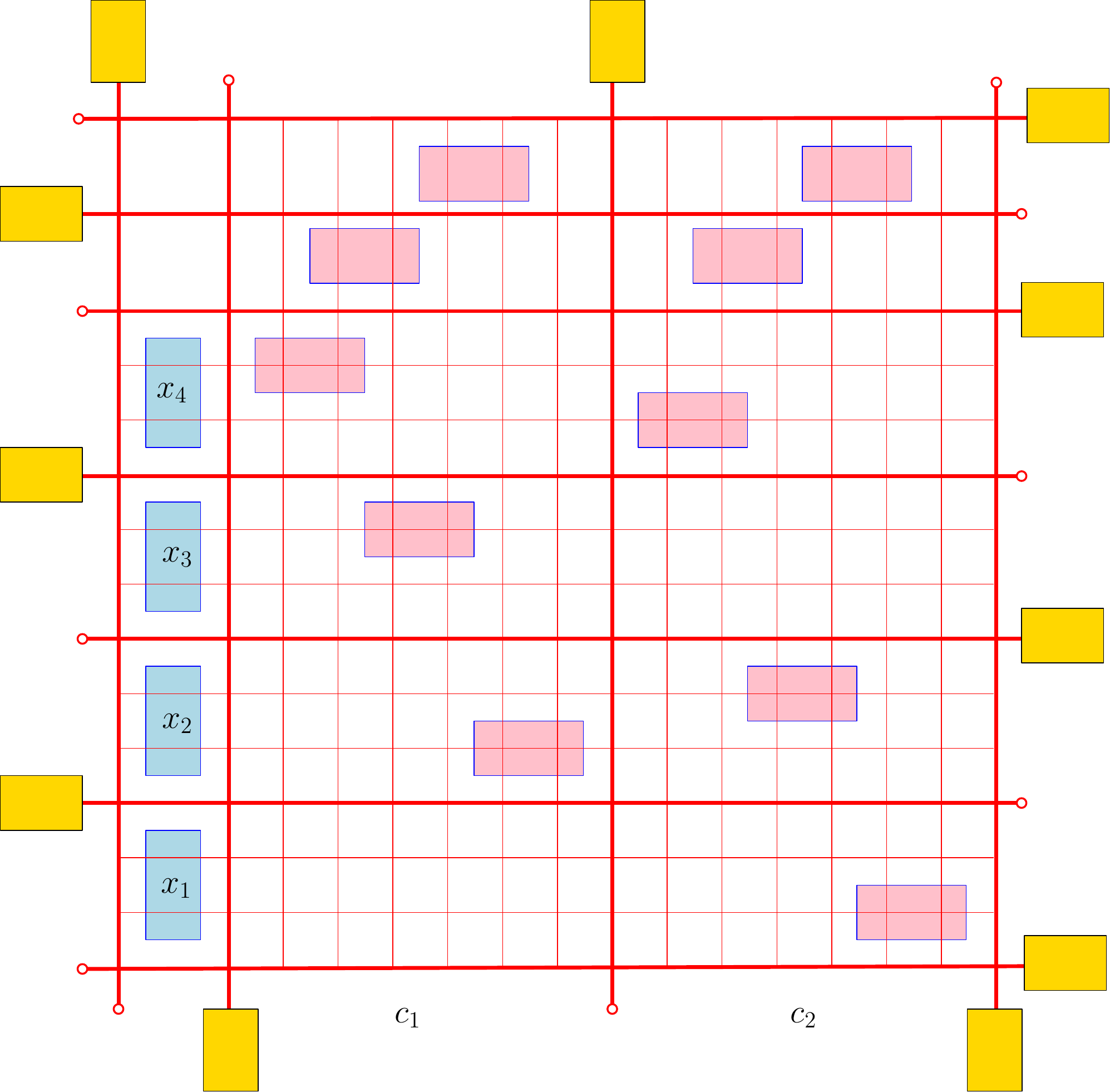}
 	\caption{An example for the rooms and corridors for $\Phi=(x_2\vee \overline{x}_3\vee\overline{x}_4)\wedge (x_1\vee \overline{x}_2\vee x_4)$. The rooms corresponding to variables are in blue. The rooms corresponding to clauses (five rooms per clause) are in pink. The corridors are shown in red, some of which are connected to additional rooms (shown in yellow).
The polygon is composed of the union of all the rooms and corridors.\label{fig:reduction}}
 \end{figure*}

\smallskip\noindent\textbf{Variable rooms.} For each variable $x_i$, $i=1,\ldots , n$, create one room:
$[\frac12,\frac{3}{2}]\times [3(i-1)+\frac12, 3i-\frac12]$. Note that all rooms are to the
left of the line $x=2$.

\smallskip\noindent\textbf{Clause rooms.} For each clause $c_j$, $j=1,\ldots, m$, create five rooms.
All five rooms have size $2\times 1$ and lie between the lines $x=7(j-1)+2$ and $x=7j+2$.
Three out of five rooms are aligned with the variable rooms. Suppose $c_j$ contains the variables
$x_i$, $x_{i'}$, and $x_{i''}$, where $i<i'<i''$. If $x_i$ is nonnegated, then create the room
$[7(j-1)+\frac12, 7(j-1)+\frac{5}{2}]\times [3(i-1)+\frac12,3i-\frac{3}{2}]$; otherwise
create the room $[7(j-1)+\frac12, 7(j-1)+\frac{5}{2}]\times [3(i-1)+\frac{3}{2},3i-\frac12]$.
We create a room for $x_{i'}$ (resp., $x_{i''}$) analogously, shifted by a horizontal vector
$(0,2)$ (resp;., $(0,4)$). Note that the $x$-projections of these rectangles do not overlap.
Two additional rooms lie above the variable rooms:
$[7(j-1)+\frac{3}2, 7(j-1)+\frac{7}{2}]\times [2n+\frac12,2n+\frac{3}{2}]$ and
$[7(j-1)+\frac{7}2, 7(j-1)+\frac{11}{2}]\times [2n+\frac{5}{2},2n+\frac{7}{2}]$.

\smallskip\noindent\textbf{Corridors and separator gadgets.}
Create narrow corridors along the vertical lines $x=0,2,3,\ldots, 7m$ and
horizontal lines $y=0,1,2,\ldots , 3n$, $y=3n+2$, and $y=3n+4$. Add rectangular rooms
of area $2$ at one end of some of the corridors. Specifically, we add rooms to the corridors at
$x=0$ and $x=7j+2$ for $j=0,1,\ldots , m$ alternately at the top and bottom endpoints;
and similarly for the corridors at $y=3i$ for $i=0,1,\ldots, n$, $y=3n+2$, and $y=3n+4$,
alternately at the left and right endpoints.
Altogether, $m+n+5$ corridors have rooms at their endpoints.

Finally, we set the parameter $k=3m+2n+5$. This completes the description of an instance corresponding to the Boolean formula $\Phi$.

\smallskip\noindent\textbf{Equivalence.}
Let $\tau:x_i\rightarrow \{\texttt{true},\texttt{false}\}$ be a satisfying truth assignment for $\Phi$.
We show that $P$ can be subdivided by $k$ lasers into regions of area at most $2$.
Place lasers at all horizontal and vertical lines that have additional rooms at their endpoints; this requires $m+n+5$ lasers.
These lasers subdivide $P$ into subpolygons that each intersect at most one room.
For $i=1,\ldots ,n$, if $\tau(x_i)=\texttt{true}$, then place a horizontal laser at $y=3(i-1)+1$ (along the bottom corridor touching room for $x_i$), otherwise at $y=3(i-1)+2$ (along the top corridor touching room for $x_i$).
These lasers split each variable room into two rectangles of area $\frac12$ and $\frac32$.
For $j=1,\ldots, m$, we place two vertical lasers that subdivide the rooms associated with clause $c_j$.
Since $\tau$ is a satisfying truth assignment, the rooms corresponding to true literals are already
split by horizontal lasers. As can easily be checked, the remaining (at most 4) rooms can be split
using two vertical lasers. Now $P$ is subdivided into pieces that each intersect at most one room,
and contains at most $1.5$ area of each room. Since the corridors are narrow,
the area of each piece is less than 2, as required.
We have used $n$ horizontal lasers for the variables, and $2m$ vertical lasers for clauses.
Overall, we have used $(m+n+5)+n+2m=3m+2n+5$ lasers.\\
Suppose now that $k=3m+2n+5$ lasers can subdivide $P$ into polygons of area at most 2.
We show that $\Phi$ is satisfiable. The area of each room is about 2, so they each intersect
at least one laser. Each variable room requires at least one laser; and the $n$ variable rooms
jointly require $n$ lasers (as no laser can intersects two variable rooms).
Each clause is associated with two rooms above the line $y=3n$; which jointly require two lasers.
Overall these rooms require $2m$ lasers.

Note that a laser that intersects a clause rooms above $y=3n$
or a variable room cannot intersect any room at the end of corridors.
We are left with at most $k-(n+2m)=m+n+5$ lasers to split these rooms.
Since we have precisely $m+n+5$ rooms at the end of the corridors, and no laser
can intersect two such rooms, there is a unique laser intersecting each of these rooms.
As argued above, for $i=1,\ldots , n$, the room associated with $x_i$ intersects only one laser.
If this laser intersects the corridor at $y=3(i-1)+1$, then let $\tau(x_i)=\texttt{true}$,
otherwise $\tau(x_i)=\texttt{false}$. For $j=1,\ldots , m$, there are two lasers that intersect
the two rooms associated with $c_j$ above $y=3n$. These two lasers cannot intersect all three
rooms associated with $c_j$ below $y=3n$. Consequently, at least one of these rooms intersects
a laser coming from a variable room. Hence each clause contains a true literal, and $\Phi$ is satisfiable.
\end{proof}

\begin{theorem}\label{thm:hardcircle}
In polygons with holes, both \mc and \mlc are NP-hard (with or without the axis-aligned lasers restriction).
\end{theorem}
\begin{proof}
The hardness reduction from 3SAT in the proof of Theorem~\ref{thm:hardness1} goes through if we replace all the $1\times 2$ and $2\times 1$ rectangles in the variable and clause gadgets with axis-aligned squares of size $\frac32\times \frac32$, and set the desired inradius to be $\delta= \frac58$. These parameters ensure that if a laser along a grid line intersects $\frac32\times \frac32$, it subdivides it into two regtangles, each of inradius at most $\frac58$.
%
%(The rooms (in gold color) along the boundary of $B$ are not needed) \rathish{In that case the number of lasers required will be less. Do we want to specify it? If we don't say anything about the gold color boxes, then nothing changes}.\csaba{OK, let's keep the gold boxes.}
\end{proof}

\begin{theorem}\label{thm:hardness2}
In polygons with holes, both \md and \mld are NP-hard (with or without the axis-aligned lasers restriction).
\end{theorem}
\begin{proof}
We reduce the problem from 3-SAT. Let $\Phi$ be a boolean formula in 3CNF with $m$ clause  $c_1,\ldots ,c_m$, and $n$ variables $x_1,\ldots , x_n$. We construct a polygon $P$ with holes and an integer $k$ such that $\Phi$ is satisfiable if and only if $P$ can be subdivided into regions of diameter at most \d using $k$ lasers (where \d=1.103 and $k$ is polynomial in $m$ and $n$ as specified below).

Similarly to the proof of Theorem~\ref{thm:hardness1}, we start with a network of horizontal and vertical corridors in a bounding box $B=[0,7m+2]\times [0,3n+4]$, which we call a \emph{grid}. The holes of the polygon are the grid cells.
Instead of rectangular ``rooms''  of a certain area, we add \emph{side-corridors} of specified shapes and lengths.
The side-corridors are narrow, and do not introduce new holes; but they have an impact on the diameter of polygonal pieces.
See Fig.~\ref{fig:reduction2} for an overview.

 \begin{figure*}[htbp]
 \centering
 	\includegraphics[width=0.60\textwidth]{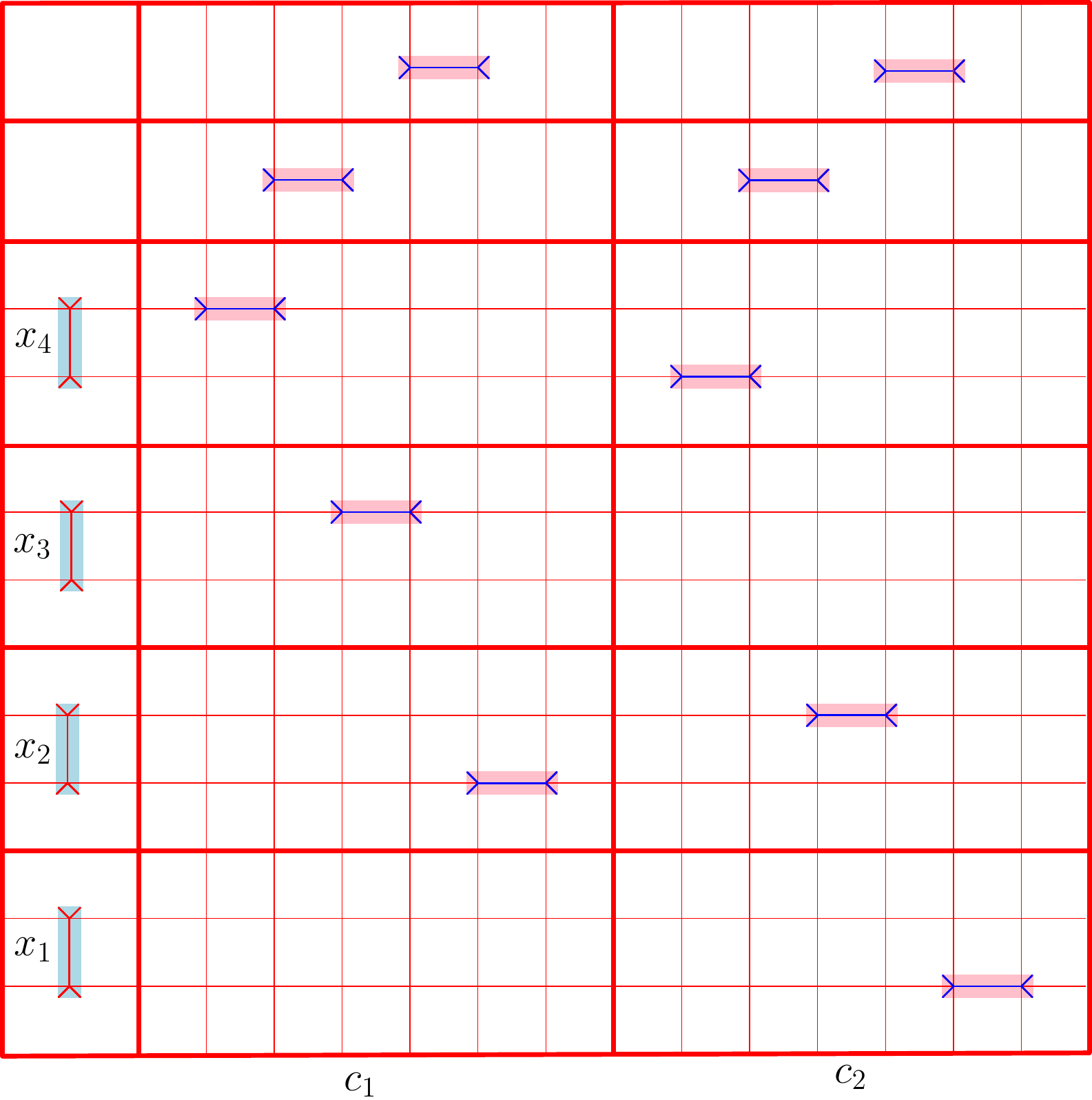}
 	\caption{An example for the grid of corridors for $\Phi=(x_2\vee \overline{x}_3\vee\overline{x}_4)\wedge (x_1\vee \overline{x}_2\vee x_4)$ (in red). The clause gadgets (five gadgets per clause) are in shaded pink. The side corridors are not shown; they are detailed in Fig.~\ref{fig:diameter-room2}.}\label{fig:reduction2}
 \end{figure*}

Each $1\times 2$ (resp., $2\times 1$) room in the proof of Theorem~\ref{thm:hardness1} is replaced by
a concentric $(1+2a)\times (2a)$ (resp., $(2a)\times (1+2a)$) bounding box $b$, where $a=0.05$.
See Fig.~\ref{fig:diameter-room1}.
These bounding boxes are not contained in the polygon; in each bounding box, \P contains
a unit-length axis-parallel corridor and four \emph{spikes} that connect
the four corners of the bounding box to the unit segment
(in particular, the diameter of the unit segment and the four spikes is $\diam(b)$).

Let the diameter threshold be $\delta:=1.103$, where $\diam(b)=\sqrt{(1+2a)^2+(2a)^2} = \sqrt{1.22}> 1.104$.
Each box $b$ can be split into pieces of diameter less than $\delta$ by a laser along a grid line passing through $b$.
Therefore, we can prove the correctness of the reduction similarly to the proof of Theorem~\ref{thm:hardness1}.
However, we need to control how the grid is subdivided into pieces of diameter at most $\delta$
(independent of the truth assignment of the variables). This is the primary role of the side-corridors
that we discuss in the remainder of the proof.

 \begin{figure}[htbp]
 \centering
 	\includegraphics[width=0.75\textwidth]{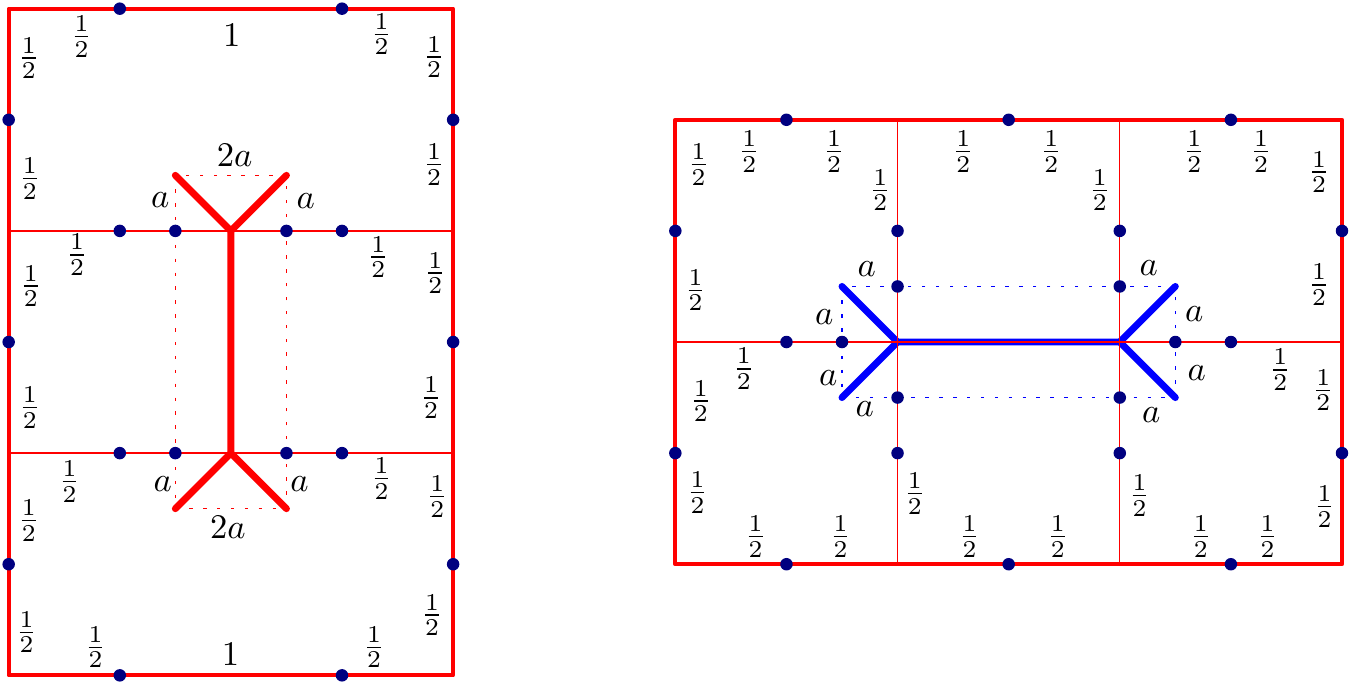}
 	\caption{A variable gadget (left) and a clause gadget (right), with four spikes of length $\sqrt{2}a$. Each dot indicates a pair of doors on opposite sides of corridors, which are shown in detail in Fig.~\ref{fig:diameter-room2}.\label{fig:diameter-room1}}
 \end{figure}

We attach a set $\mathcal{S}$ of side-corridors to the grid as follows:
\benum
\item For every $s\in \mathcal{S}$, $\diam(s)= \sqrt{5/4}\approx 1.118>\delta$, hence $s$ requires at least one laser.
\item Every $s\in \mathcal{S}$ is connected the grid by a very short and narrow corridor, called the \emph{door} of $s$; a laser through this door can split $s$ into two pieces of diameter less than $\d$.
\item $\mathcal{S}$ is partitioned into pairs $\{s_1,s_2\}\subset \mathcal{S}$ such that the doors of $s_1$ and $s_2$ are located symmetrically on opposite sides of a grid corridor, and a laser can split both $s_1$ and $s_2$ into pieces of diameter at most $\d$ iff $\{s_1,s_2\}$ is a pair. (Fig.~\ref{fig:diameter-room2}(middle)--(right) show side-corridors in two adjacent grid cells on opposite sides of a corridor.)
\item Therefore, every optimal solution contains a laser for each pair of side-corridors in $\mathcal{S}$.
\item There is a pair of doors at every intersection between a main corridor (of the grid) and the bounding box of a gadget.
\item The number $k$ of lasers equals $n+2m$ (one per variable gadget, two per clause gadget) plus half the number of side-corridors, which is proportional to the total length $O(mn)$ of the grid.
\eenum

 \begin{figure}[htbp]
 \centering
 	\includegraphics[width=0.98\textwidth]{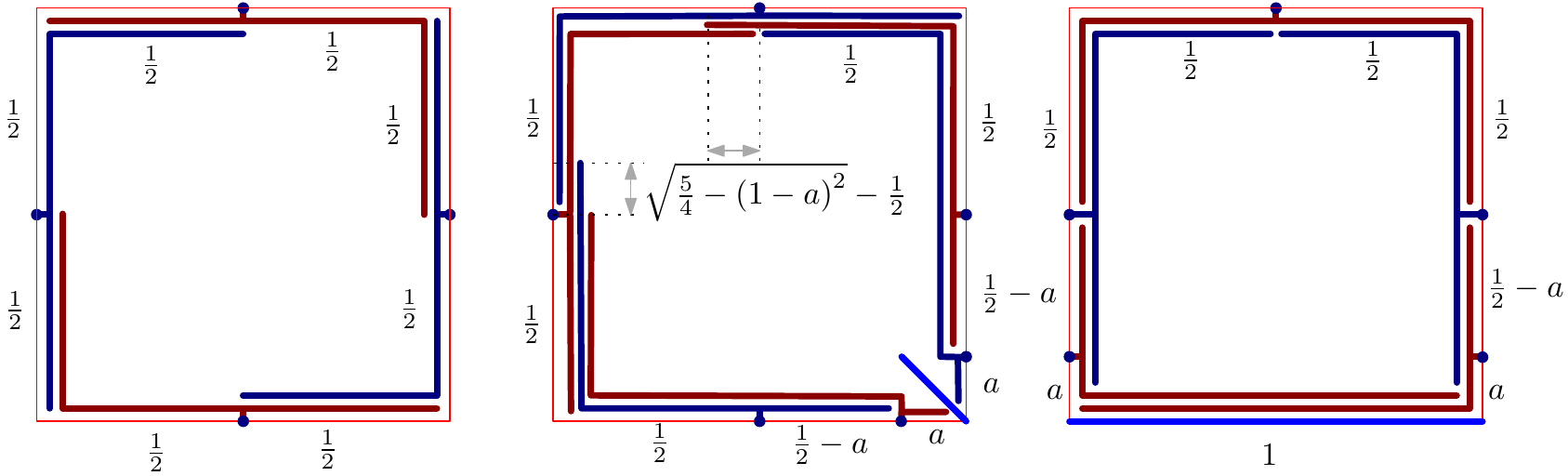}
 	\caption{Side-corridors for three types of holes of the grid. Left: A hole not adjacent to any gadget. Middle: a hole adjacent to a corner of a gadget (with a spike in the lower right corner). Right: a hole adjacent to the side of a gadget (the unit-length corridor of a gadget is the bottom edge, highlighted in blue).
 %\jie{what is a hole? Also what is the blue segment in the picture?}
 %\csaba{Hole of the grid. The blue segments are a spike and a unit-length corridor of a gadget}
 \label{fig:diameter-room2}}
 \end{figure}

Figure~\ref{fig:diameter-room1} indicates the locations of the doors to matching pairs of side-corridors in a grid;
and Figure~\ref{fig:diameter-room2} shows the specific shapes of these side-corridors in a grid cell.
Lasers placed at every pair of side-corridors decompose the polygon into cells of diameter at most $\delta$,
and the bounding boxes of the gadgets boxes. It follows that the reduction in the proof of Theorem~\ref{thm:hardness1}
goes though, the decision problem whether $P$ can be subdivided into pieces of diameter $\delta$ with a given
number of lasers is NP-hard.
\end{proof}

\section{Decomposition Algorithms for Simple Polygons}
\label{sec:general-cell}

In this section, we present approximation results for decomposing a simple polygon \P 
by lasers of arbitrary orientations (recall that $n$ denotes the total number of vertices of \P and $r$ is the number of reflex vertices).
We describe an $O(\log r)$-approximation for \mla (Section~\ref{sec:minlaser-area}), a bi-criteria algorithm for diameter (Section~\ref{sec:minlaser-diameter}), and a $O(1)$-approximation for \md (Section~\ref{sec:k-laser-mindiameter}).

\subsection{\mla}
\label{sec:minlaser-area}

Given a simple polygon \P and a threshold \d, we wish to find the minimum number of lasers that subdivide \P into pieces, each of area at most 1. We start with the easy $O(1)$-approximation in the special case when \P is a convex polygon.

\begin{lemma}
\label{lem:Convex-ConvexArea}
For every convex polygon $P$, we can find a set of $k=O(\sqrt{\area(P)})$ lasers that subdivide $P$ into pieces, each of area at most 1, in $O(k+n)$ time. Every decomposition into pieces of area at most 1 requires $\Omega(\sqrt{\area(P)})$ lasers.
\end{lemma}
\begin{proof}
	For the lower bound, notice that the arrangement of $\ell$ lines has $O(\ell^2)$ faces, and so $\ell$ lasers decompose $P$ into $O(\ell^2)$ cells. By the pigeonhole principle, the area of the largest piece is $\Omega(\area(P)/\ell^2)\leq 1$. Hence $\ell\geq \Omega(\sqrt{\area(P)})$, as claimed.
	
	For the upper bound, recall that by John's theorem~\cite{Bar02,John48}, $P$ contains an ellipse $E$ such that $E\subset P\subset 2E$, where $2E$ is obtained from $E$ by central dilation of ratio 2. Let $B$ be a bounding box of $P$ (i.e., a rectangle of arbitrary orientation that contains $P$) of minimum area. Then $\area(E)\leq \area(P)\leq \area(B)\leq \frac{4}{\pi}\,\area(2E)=O(\area(E))$, hence $\area(B)\leq O(\area(P))$, and $B$ can be computed in $O(n)$ time~\cite{FreemanS75,Toussaint14}. Assume, w.l.o.g., that $B$ is axis-aligned with the lower-left corner at the origin. For every $m\in \mathbb{N}$, we can decompose the rectangle $B$ into $m^2$ congruent rectangles with $2(m-1)$ axis-parallel lasers: $m-1$ equally spaces horizontal (resp.,  vertical) lasers. If $m=\lceil \sqrt{\area(B)}\rceil=O(\sqrt{\area(P)})$, then the area of each piece is at most $\area(B)/m\leq 1$. These $\ell=2m-2$ lasers decompose $P$ into pieces of area at most $1$, as well, as required.
\end{proof}

\paragraph{Overview.}
We give a brief overview of our approximation algorithm for a simple polygon \P. The basic idea is to decompose \P into convex pieces, and use Lemma~\ref{lem:Convex-ConvexArea} to further decompose each convex piece.
There are two problems with this na\"{\i}ve approach: (1) a laser in an optimal solution may intersect several convex pieces (i.e., the sum of lower bounds for the convex pieces is not a global lower bound); and (2) the lasers used for a convex decomposition are not accounted for.
We modify the basic approach to address both of these problems.

We use the Hershberger--Suri triangulation (as a convex subdivision). For a simple polygon \P with $n$ vertices, Hershberger and Suri~\cite{suri} construct a Steiner triangulation into $O(n)$ triangles such that every chord of $P$ intersects $O(\log n)$ triangles.
We can modify their construction to produce a Steiner decomposition into a set $\mathcal{C}$ of convex cells (rather than triangles) such that each laser intersects $O(\log r)$ convex cells, where $r$ is the number of reflex vertices of \P. Thus, each laser of $\opt$ can help partition $O(\log r)$ convex cells; this factor dominates the approximation ratio of our algorithm.

A convex cell $C\in \mathcal{C}$ is \emph{large} if $\area(C)>1$, otherwise it is \emph{small}. We decompose each large convex cell using Lemma~\ref{lem:Convex-ConvexArea}. We can afford to place $O(1)$ lasers along the boundary of a large cell. We cannot afford to place lasers on the boundaries of all small cells. If we do not separate the small cells, however, they could merge into a large (nonconvex) region, so we need \emph{some} separation between them. In the algorithm below, we construct such separators recursively by carefully unrefining the Hershberger--Suri triangulation. The unrefined subdivision is no longer a triangulation, but we maintain the properties that
(i) each cell is bounded by $O(1)$ lasers within each pseudotriangle (and an arbitrary number of consecutive edges of $P$), and
(ii) every chord of $P$ intersects $O(\log n)$ cells.

\begin{figure*}[h]
\centering
\includegraphics[width=0.95\textwidth]{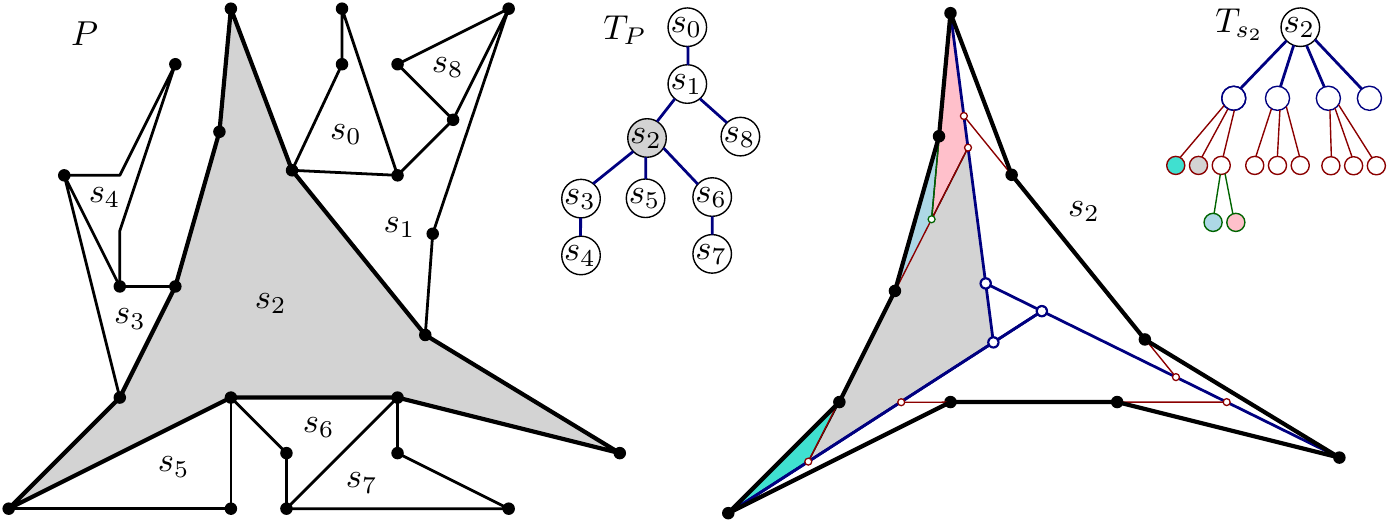}
\caption{Left: A simple polygon $P$, decomposed into pseudotriangles, and the dual graph $T_P$.
Right: A pseudotriangle $s_2$ is recursively subdivided into Steiner triangles, with recursion tree $T_{s_2}$. }
\label{fig:pseudotriangle}
\end{figure*}

\paragraph{Basic properties of the Hershberger--Suri triangulation.}
Given a simple polygon $P$ with $n$ vertices, Hershberger and Suri~\cite{suri} construct a Steiner-triangulation in two phases (see Fig.~\ref{fig:pseudotriangle} for an example): First, they subdivide $P$ into $O(n)$ pseudotriangles (i.e., simple polygons with precisely three convex vertices) using $O(n)$ noncrossing diagonals of $P$; and then subdivide each pseudotriangle into Steiner triangles. The runtime of their algorithm, as well as the number of Steiner triangles, is $O(n)$. Let $\mathcal{S}$ denote the set of pseudotriangles produced in the first phase; and let $T_P$ be the \emph{dual tree} of the pseudotriangles, in which each node corresponds to a pseudotriangle, and two nodes are adjacent if and only if the corresponding pseudotriangles share an edge (a diagonal of $P$).
Note that the degree of $T_P$ is not bounded by a constant (it is bounded by $n$), as a pseudotriangle may be adjacent to arbitrarily many other pseudotriangles.
We consider $T_P$ to be a rooted tree, rooted at an arbitrary pseudotriangle.
Then every nonroot pseudotriangle $s$ in $\mathcal{S}$ has a unique edge incident to the parent of $s$;
we call this edge the \emph{parent edge of~$s$}.

Hershberger and Suri subdivide each pseudotriangle $s\in S$ recursively: In each step, they use $O(1)$ line segments to subdivide a pseudotriangle into $O(1)$ pseudotriangles, which are further subdivided recursively until they obtain triangles. Let us denote by $T_s$ the recursion tree for $s$.
Each vertex $v\in T_s$ represents a region $R_v\subset s$: The root of $T_s$ represents $s$, and the leaves represent the Steiner triangles in $s$. The recursive subdivision maintains the following two properties:
(a) Every edge of $s$ is incident to a unique region in each level of $T_s$,
(b) For each node $v\in T_s$, the boundary between $R_v$ and $s\setminus R_v$ is a  polyline with $O(1)$ edges (that is, $R_v$ is bounded by $O(1)$ line segments inside $s$, and some sequence of consecutive edges of $s$).

\smallskip\noindent\textbf{Algorithm.}
We are ready to present an approximation algorithm for \mla.
Given a simple polygon $Q$, we begin by computing the Hershberger--Suri triangulation, the pseudotriangles $\mathcal{S}$, the dual tree $T_P$, and a recursion tree $T_s$ for each pseudotriangle $s\in \mathcal{S}$. We then process the pseudotriangles in a bottom-up traversal of $T_P$.

%Within each pseudotriangle $s\in \mathcal{S}$, we unrefine the Steiner triangulation by merging some of the cells into one cells (the resulting larger cells need not be triangular or convex). Each node $v\in T_s$ corresponds to a region $R_v\subseteq s$. However, if we do not place lasers along the edges of $s$, then cell of the overall subdivision extends beyond the boundary of $s$, that is $R_v\subseteq s$ is part of some larger cell $\widehat{R}_v\subseteq P$, where $\area(\widehat{R}_v) = \area(R_v) + \sum_{e}w(e)$, where the summation is over all edges of $s$ on the boundary of $R_v$, and $w(e)$ denotes the area of the cell on the opposite side of $e$. In the course of the algorithm, we compute nonnegative \emph{weights} for all edges of the pseudotriangles in $\mathcal{S}$. The weights are used for bookkeeping purposes. The edges of $P$ have zero weight. In a bottom-up traversal of $T_P$, we may assume that the \emph{weights} $w(e)$ have already been computed for all but the parent edge of $s$.

%For the parent edge of $s$, the weight $w(e)$ has not been computed when we process $s$. (The weight $w(e)$ is determined with we have computed the unrefined subdivision of $s$, and $w(e)$ will be the area of the unrefined cell in $s$ adjacent to the parent edge.)

Within each pseudotriangle $s\in \mathcal{S}$, we unrefine the Steiner triangulation of $s$ by merging some of the cells into one cell (the resulting larger cells need not be triangular or convex). Initially, each node $v\in T_s$ corresponds to a region $R_v\subseteq s$. However, if we do not place lasers along the edges of $s$, then $R_v$ may be adjacent to (and merged with) other cells that are outside the pseudotriangle $s$, along the boundary of $s$. Since we have an upper bound on the total area of each cell in the final decomposition, we need to keep track of the area of the region on both sides of an edge of the pseudo-triangulation. In the course of unrefinement algorithm for all $s\in \mathcal{S}$, we compute nonnegative \emph{weights} $w(\cdot)$ for all edges of the pseudotriangulation. The weights are used for bookkeeping purposes. Specifically, the edges of $P$ have zero weight. In a bottom-up traversal of $T_P$,
when we start processing a pseudotriangle $s$, the \emph{weights} $w(e)$ have already been computed for all edges of the pseudo-triangle $s$ except the parent edge of $s$. The weight $w(e)$ for the parent edge $e$ of $s$ is determined when we have computed the unrefined subdivision of $s$;
and $w(e)$ will be the area of the unrefined cell in $s$ adjacent to the parent edge.
A node $v\in T_s$ initially corresponds to a region $R_v$ within the pseudotriangle $s$, but in the final decomposition of $P$, the node is part of some larger cell $\widehat{R}_v\subseteq P$, with $\area(\widehat{R}_v) = \area(R_v) + \sum_{e}w(e)$, where the summation is over all edges of $s$ on the boundary of $R_v$, and $w(e)$ denotes the area of the cell on the opposite side of $e$.

As the weight of the parent edge is not available yet when we unrefine $s$, we modify the recursion tree $T_s$ as follows: We choose the root to be the leaf $v_0\in T_s$ adjacent to the parent edge of $s$, and reverse the parent-child relation on all edges of $T_s$ along the $s$-$v_0$ path. We denote the modified recursion tree $T'_s$ (Fig.~\ref{fig:pseudotriangle2} (left)).
For all nodes $v$ along the $s$-$v_0$ path (including $s$ and $v_0$), we redefine the corresponding regions of the nodes in $T'_s$ as follows. We denote by $R_v(T_s)$ and $R_v(T'_s)$ the regions corresponding to node $v$ in trees $T_s$ and $T'_s$, respectively. We set $R_{v_0}(T'_s):=s$ and for all other nodes $v$ along the $s$-$v_0$ path (including $s$), we set $R_{v}(T'_s):=s\setminus R_u(T_s)$, where $u$ is the parent of $v$ in $T'_s$.
%As we only consider tree $T'_s$
With a slight abuse of notation, we set $R_v = R_v(T'_s)$ for all $v\in T'_s$ for the remainder of the algorithm.
Note that $\area(R_{v})$ monotonically decreases with the depth in $T'_S$.

%For all nodes $v$ along the $s$-$v_0$ path (including $s$ and $v_0$), we redefine the corresponding region by setting $R_{v_0}:=s$, and for all other nodes along the path, $R_v:=s\setminus R_u$, where $u$ is the parent of $v$ in $T_s'$.

%In the course of the algorithm, we compute nonnegative \emph{weights} for all edges of the pseudotriangles in $\mathcal{S}$. The weights are used for bookkeeping purposes. The edges of $P$ always have zero weight. Two adjacent pseudotriangles are in parent-child relation in the dual tree $T_P$. For the edge $e$ between $s$ (parent) and $s'$ (child), the weight $w(e)$ represents an area that $s'$ and its descendants ``contribute'' to the parent $s$ in the following sense: If the edge $e$ is not present in our current unrefined) subdivision of $P$, then $e$ lies in a some cell $C$ that lies on both sides of $e$; the weight $w(e)$ is the area the part of cell $C$ that lies in $s'$ and its descendants. When the algorithm processes $s$, we can compute $\area(C\cap s)$ directly, and the remaining area $\area(C\setminus s)$ of $C$ is given by the sum of weights of the edges between $C$ and all children of $s$. In particular, every node of $v$ of the recursion tree $T_s$ corresponds to a region $R_v$; then $\area(R_v)$ equals the sum of $\area(R\cap s)$ plus the summation of the weights of all edges of $s$ on the boundary of $R_v$.

In a bottom-up traversal of $T_P$, consider every $s\in \mathcal{S}$.
%Assume that the weights of the edges between $s$ and its children have already been computed.
We proceed with two phases (see Fig.~\ref{fig:pseudotriangle2} for an example).

\begin{figure}[htbp]
\centering
\includegraphics[width=0.95\textwidth]{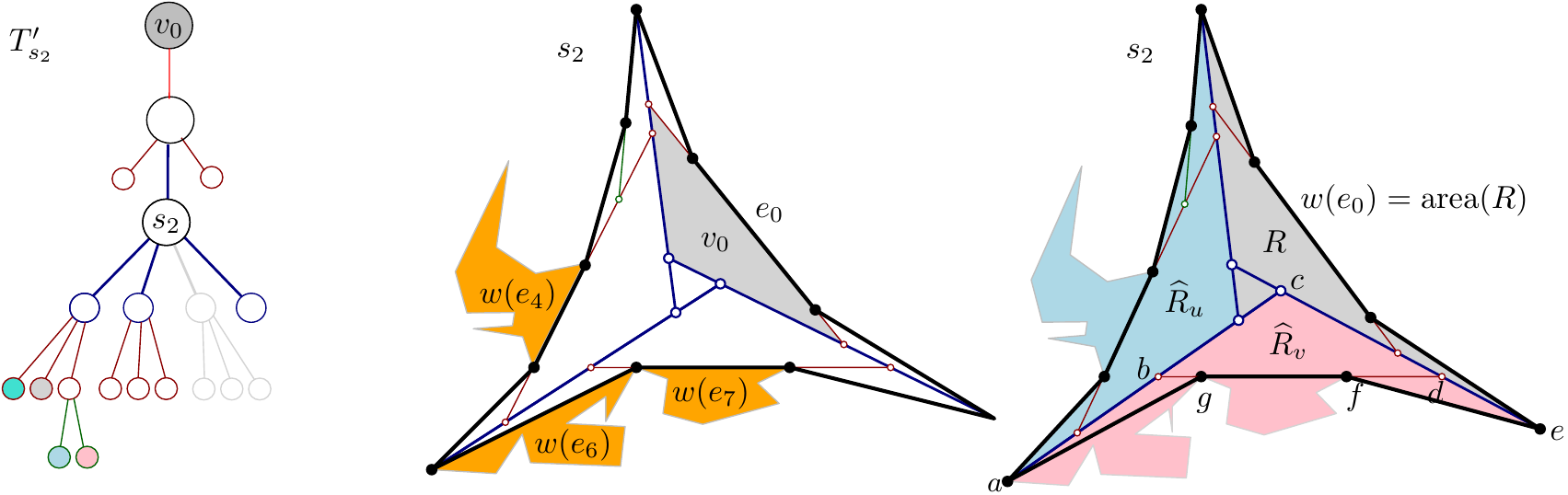}
\caption{Left: The modified recursion tree $T_{s_2}'$.
Middle: pseudotriangle $s_2$ with the initial Steiner triangulation, edge weights representing the areas of adjacent regions in the descendants of $s_2$, and the parent edge $e_0$ of $s_2$.
Right: The unrefined subdivision of $s_2$ into $R_u$, $R_v$, and $R$; larger cells $\widehat{R}_u$ and $\widehat{R}_v$ (blue and pink),
and the weight $w(e_0)=\area(R)$ of the parent edge of $s_2$ (gray).}
\label{fig:pseudotriangle2}
\end{figure}

\noindent Phase~1 of the algorithm is an unrefinement process, that successively merges small cells of the Hershberger--Suri triangulation (no lasers are involved).  We initialize three variables:
    \[R:=s,\hspace{6mm} T:=T'_s,\hspace{6mm} U_s:=\emptyset,\]
where $R\subseteq s$ is the region yet to be handled, $T$ is a subtree of $T_s'$ corresponding to the region $R$,
and $U_s$ is the set of interior-disjoint faces in $s$ produced by the unrefinement process.
While $\area(R)>1$, do the following:
\bitem
\item Find a lowest node $v\in T$ for which $\area(\widehat{R}_v)>1$,
\item Set $U_s:=U_s\cup \{\widehat{R}_v\}$,
\item Set $R:=R\setminus R_v$,
\item Delete the subtree rooted at $v$ from $T$, and
\item For all ancestors $u$ of $v$, set $\widehat{R}_u:=\widehat{R}_u\setminus \widehat{R}_v$.
\eitem
When the while loop ends, define the weight of the parent edge of $s$ to be $\area(R)$.

\noindent Phase~2 of the algorithm positions lasers in a pseudotriangle $s$ as follows.\\
For every region $\widehat{R}_v\in U_s$, do:
\bitem
\item Step~1. Place lasers along all edges of the boundary between $\widehat{R}_v$ and $s\setminus \widehat{R}_v$, and the boundaries between $R_{v}$ and $R_{v'}$ for all children $v'$ of $v$. For example in Fig.~\ref{fig:pseudotriangle2} (right), two lasers are placed along the edges $(a,c)$ and $(c,e)$ that disconnect $\widehat{R}_v$ from $s_2$. Also, a laser that is placed along edge $(b,d)$ that separates the children of $R_{v}$. %\rathishI{@Csaba, please check if the previous line is fine or not.} \csaba{It's fine.}

\item Step~2. If $\area(R_v)\geq 1$ (which means $R_v$ has not merged with any other region in Phase~1, i.e., $\widehat{R}_v=R_v$ hence $\widehat{R}_v$ is convex), subdivide $\widehat{R}_v$ by $\Theta(\sqrt{\area(R_v)})$ lasers according to Lemma~\ref{lem:Convex-ConvexArea}.
%
%For example in Fig.~\ref{fig:pseudotriangle2} (right), if the area of triangle $R_v = \bigtriangleup cbd$ has area more than $1$, then we subdivide $\bigtriangleup cbd$ by $\Theta(\sqrt{\area(\bigtriangleup cbd)})$ lasers. %\rathishI{@Csaba, please check if the previous line is fine or not.}\csaba{No, that's incorrect; none of the regions happens to be convex in the example in Fig.~\ref{fig:pseudotriangle2} (right). I've revised the paragraph.}
\eitem
This completes the description of our algorithm.

\begin{theorem}\label{thm:minlaser-area}
Let \P be a simple polygon with $n$ vertices, and let $k^*$ be the minimum number of lasers that subdivide \P into pieces of area at most $1$. We can find an integer $k$ with $k^*\leq k\leq O(k^*\log n)$ in $O(n)$ time, and a set of $k$ lasers that subdivide \P into pieces of area at most $1$ in output-sensitive $O(k+n)$ time.
\end{theorem}
\begin{proof}
%val commented out since it's in the theorem statement. Let $k^*$ be the number of laser cuts in an optimal solution $\opt$, and let $k$ be the number of lasers computed by the algorithm.
Phase~1 of our algorithm (unrefinement) subdivides each pseudotriangle $s\in \mathcal{S}$ into regions such that each region corresponds to a subtree rooted at some node $v$ of the recursion tree $T_s'$. Node $v$ corresponds to a region $R_v\subset s$, and a possibly larger region $\widehat{R}_v\subset P$ which is the union of $R_v$ and adjacent regions in the descendant pseudotriangles of $s$ adjacent to $R_v$. Phase~1 of the algorithm ensures that $\area(\widehat{R}_v)> 1$ (therefore, $\widehat{R}_v$ must intersect at least one laser in $\opt$), but for all children $v'$ of $v$ in $T'_s$, we have $\area(\widehat{R}_{v'})\leq 1$.

In Step~1, the algorithm uses $O(1)$ lasers for each $v\in U_s$ to separate $\widehat{R}_v$ from $s\setminus \widehat{R}_v$.
Recall that the recursion tree $T_s$ has bounded degree. Consequently, we use $O(1)$ lasers to separate $\widehat{R}_{v'}$ from $\widehat{R}_v\setminus \widehat{R}_{v'}$ for all children $v'$ of $v$. These polylines subdivide $\widehat{R}_{v'}$ into smaller regions of area at most $1$. Overall, we have used $O(1)$ lasers for each of these nodes $v\in T_s'$, $s\in \mathcal{S}$.
Note that each region $\widehat{R}_v$ is the union of triangles from the Hershberger--Suri Steiner triangulation,
and so each laser in $\opt$ intersects $O(\log n)$ such regions. Consequently, we use $O(k^* \, \log n)$ lasers in Step~1.

Finally, consider the lasers used in Step~2 for subdividing the triangles $t\in T$ with $\area(t)>1$.
By Lemma~\ref{lem:Convex-ConvexArea}, each such triangle intersects at least $\Omega(\sqrt{\area(t)})$
lasers in any valid solution; and conversely each laser of an optimal solution intersects $O(\log n)$ regions in $T$.
Consequently, the number of lasers uses in  Step~2 is $\sum_{t\in T}O(\sqrt{\area(t)})\leq O(k^*\, \log n)$.

It remains to show that the algorithm runs in $O(n+k)$ time.
The Hershberger-Suri Steiner triangulation can be computed in $O(n)$ time~\cite{suri}. It consists of $O(n)$ triangles, hence the combined size of the dual tree $T_p$, and the recursion trees $T_s$, $s\in \mathcal{S}$, is also $O(n)$. The unrefinement algorithm is done in a single traversal of these trees, spending $O(1)$ time at each node. For each large cell (triangle) of the Hershberger-Suri triangulation, by Lemma~\ref{lem:Convex-ConvexArea}, we can compute a minimum bounding box and the \emph{number} of lasers used by the algorithm in $O(1)$ time. Computing all $k$ lasers requires $O(k)$ additional time.
\end{proof}

\paragraph{An $O(\log r)$ Approximation for \mla in Simple Polygons.}
We can improve the approximation ratio in Theorem~\ref{thm:minlaser-area} from $O(\log n)$ to $O(\log r)$, where $r$ is the number of reflex vertices of $P$, if we replace the Hershberger--Suri triangulation with a convex decomposition.
(Hershberger and Suri decompose $P$ into \emph{triangles} to support ray shooting queries, but for our purposes a decomposition into convex cells suffices.)

Let $(v_1,\ldots , v_n)$ be the $n$ vertices of \P; assume they are in general position. %in $\mathbb{R}^2$.
Let $R$ be the set of reflex vertices of $P$. For every reflex vertex $v\in R$, the angle bisector of the interior angle at $v$ hits some edge $a_vb_v$ of $P$. Let $L=\{v,a_v,b_v:v\in R\}$, that is, $L$ is the set of all reflex vertices of $P$, and both endpoints of the edges hit by the angle bisectors of reflex angles. Clearly, $|L|\leq 3r$.

\begin{lemma}\label{lem:convex0}
There is a simple polygon $Q\subset P$ whose vertex set is $L$,
and every connected component of $P\setminus Q$ is a convex polygon.
The polygon $Q$ can be computed in $O(n\log n)$ time.
\end{lemma}
\begin{proof}
We describe an algorithm that decomposes $P$ along noncrossing diagonals into a collection $\mathcal{P}$ of convex polygons and their complement $P\setminus (\bigcup_{A\in \mathcal{P}} A)$ will be polygon $Q$. Initially $\mathcal{P}=\{P\}$ and $Q=\emptyset$.

The algorithm has two steps. In the first step, a collection $\mathcal{P}$ of convex polygons is created such that the vertex set of the complement $P\setminus (\bigcup_{A\in \mathcal{P}} A)$ is $L$. However, $P\setminus (\bigcup_{A\in \mathcal{P}} A)$ is not necessarily connected. In the second step, the connected components of $P\setminus (\bigcup_{A\in \mathcal{P}} A)$ are merged into a simple polygon $Q$ (a single connected component) with the same vertex set $L$.

\paragraph{First step:} While there is a nonconvex polygon $P'\in \mathcal{P}$,
we replace $P'$ with one or more smaller polygons in $\mathcal{P}$ as follows.
Let $v$ be a reflex vertex of $P'$. Since $P'\subset P$, vertex $v$ is also a reflex vertex of $P$. Denote by $\vec{r}_v$ the angle bisector of $P$ at $v$. Note that $\vec{r}_v$ enters the interior of $P'$ at $v$; denote by $ab$ the edge of $P'$ where $\vec{r}_v$ first exits $P'$. Let $P(v)$ be the geodesic triangle formed by the edge $ab$ and the shortest paths from  $v$ to $a$ and to $b$, respectively.
Update $\mathcal{P}$ by replacing $P'$ with the polygons in $P'\setminus P(v)$. See Figure~\ref{fig:decomposition} for an example.
In the course of the algorithm, every polygon in $\mathcal{P}$ is formed by a sequence of consecutive vertices of the input polygon $P$.

\begin{figure}[htbp]
\centering
\includegraphics[scale=0.75]{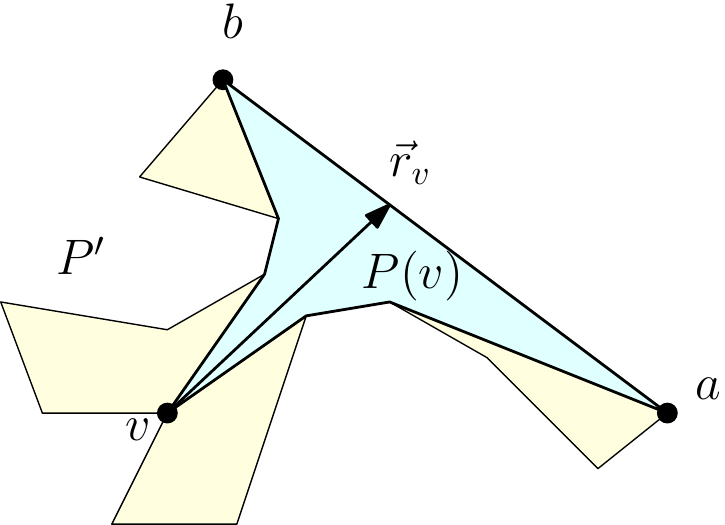}
\caption{Replace $P'$ by four polygons (in yellow), after taking out the geodesic triangle $P(v)$ where $v$ is a reflex vertex. }
\label{fig:decomposition}
\end{figure}

We claim that in each iteration of the algorithm, all vertices of $P(v)$ are in $L$. Clearly, $v$ is a reflex vertex in $P'$, hence a reflex vertex of $P$, as well. Similarly, the interior vertices of the shortest paths from $v$ to $a$ and to $b$ are reflex vertices in $P'$, hence in $P$. It remains to show that $a,b\in L$. If $ab$ is an edge of $P$, then $a,b\in L$ by the definition of $L$. Otherwise, $ab$ is a diagonal of $P$, and so it is an edge of a geodesic triangle $P(v')$ of some previous iteration of the algorithm---by induction, they are in $L$, as well. At the end of the while loop, all polygons in $\mathcal{P}$ are convex, however, the complement $P\setminus (\bigcup_{A\in \mathcal{P}} A)$ is not necessarily connected. See Figure~\ref{fig:merge2} for an illustration.

\begin{figure}[htbp]
\centering
\includegraphics[scale=0.98]{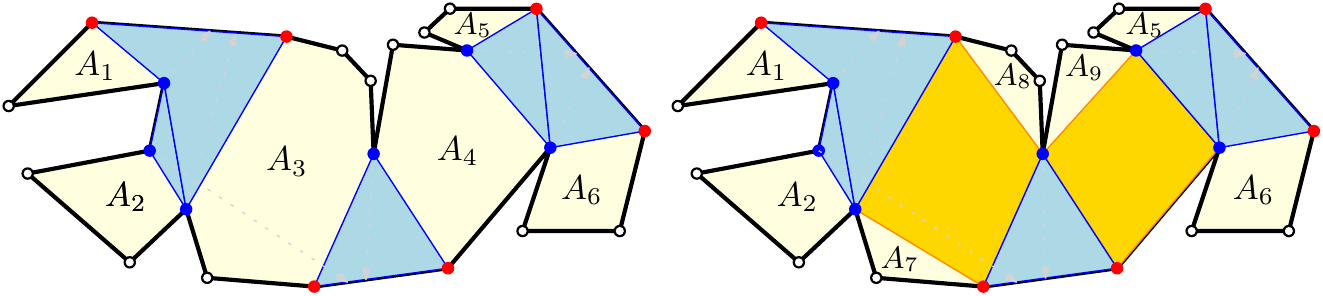}
\caption{A simple polygon $P$, the vertices in $L$ are blue (reflex) or red (hit by angle bisector).
 Left: The first step produces convex polygons $\mathcal{P}=\{A_1,\ldots , A_6\}$, but $P\setminus \bigcup_{i=1}^6 A_i$
 is disconnected.
 Right: The second step merges $P\setminus \bigcup_{i=1}^6 A_i$ into a simple polygon $Q$. As $|A_3\cap L|\geq 3$ and $|A_4\cap L|\geq 3$, the second step creates conv$(V(A_3)\cap L)$ and conv$(V(A_4)\cap L)$ (shown in deep yellow) which merges $P\setminus \bigcup_{i=1}^6 A_i$ into a single connected component $Q$.}
\label{fig:merge2}
\end{figure}

\paragraph{Second step:} While there is a (convex) polygon $P'\in \mathcal{P}$ incident to three or more vertices in $L$, we replace $P'$ with smaller polygons: In particular, let $V(P')$ be the vertex set of $P'$. If $|V(P')\cap L|\geq 3$, then replace $P'$ with the polygons in $P'\setminus \text{conv}(V(P')\cap L)$, where $\text{conv}(.)$ stands for the convex hull. In each iteration, all polygons in $\mathcal{P}$ remain convex. At the end of the while loop, every polygon in $\mathcal{P}$ is incident to exactly two vertices in $L$, and $P\setminus (\bigcup_{A\in \mathcal{P}} A)$ is a simple polygon with vertex set $L$.
%\jie{Maybe we need a figure for this paragraph? Also it seems we need to say why  $P\setminus (\bigcup_{A\in \mathcal{P}} A)$ is a simple polygon with vertex set $L$.}
%\csaba{I've added Fig.~\ref{fig:merge2} to illustrate the 2nd while loop.}
\end{proof}

\begin{lemma}\label{cor:convex3}
Every simple polygon \P on $n$ vertices, $r$ of which are reflex, can be decomposed into convex faces
such that every chord of \P intersects $O(\log r)$ faces. Such a decomposition can be computed in
$O(n \log n)$ time.
\end{lemma}
\begin{proof}
We can compute the set $L$ of up to $3r$ vertices and a simple polygon $Q\subset P$ described in Lemma~\ref{lem:convex0} in $O(n\log n)$ time. We then compute the Hershberger--Suri triangulation
for $Q$, which is a Steiner triangulation of $O(r)$ triangles such that every chord of $Q$ intersects $O(\log r)$ triangles~\cite{suri}. This triangulation of $Q$, together with the convex polygons in $P\setminus Q$, form a subdivision of $P$ into convex faces.

We claim that every chord of $P$ intersects at most $O(\log r)$ faces: at most $O(\log r)$ triangles in $Q$ and at most two convex sets in $P\setminus Q$. If a chord $\ell$ of $P$ intersects three components of $P\setminus Q$, say $C_1,C_2,C_3$ in this order, then $\ell$ crosses the boundary of $C_2$ twice, so $C_2$ must have at least two edges on the boundary between $C_2$ and $Q$. However, by Lemma~\ref{lem:convex0}, every edge of $Q$ is either an edge or a diagonal of $P$. Therefore the boundary between $Q$ and a component of $P\setminus Q$ is a single diagonal of $P$. Thus $\ell$ intersects at most two components of $P\setminus Q$; moreover $\ell\cap Q$ is a chord of $Q$, so it intersects $O(\log r)$ triangles inside $Q$.
\end{proof}
By performing the algorithm on the convex subdivision in Corollary~\ref{cor:convex3},
the approximation ratio improves to $O(\log r)$.

\begin{theorem}\label{thm:minlaser-area2}
Let \P be a simple polygon with $n$ vertices, $r$ of which are reflex, and let $k^*$ be the minimum number of lasers needed to subdivide $P$ into pieces of area at most $1$.
We can find an integer $k$ with $k^*\leq k\leq  O(k^*\log r)$ in $O(n\log n)$ time, and a set of $k$ lasers that subdivide \P into pieces of area at most $1$ in output-sensitive $O(k+n\log n)$ time.
\end{theorem}

\section{Diameter in simple polygons}
\subsection{Bi-Criteria Approximation for Diameter}\label{sec:minlaser-diameter}
For the diameter version in a simple polygon, we describe a bi-criteria approximation algorithm
(Theorem~\ref{thm:bicriteria-diameter}).
We start from deriving a lower bound for the minimum number of lasers in a decomposition
into pieces of diameter at most $\delta$ (for bi-criteria approximation algorithm we use general \d, instead of \d=1, because we will scroll over \d when using the algorithm to get an approximation for \md).

Consider the infinite set of vertical lines, $\mathcal{L}_V$, evenly spaced with separation $\delta$;
that is, $\mathcal{L}_V=\{x=i\delta: i\in \mathbb{Z}\}$.
The lines in $\mathcal{L}_V$ decompose $P$ into a set $\mathcal{P}_V$ of simple polygons,
that we call \emph{cells}. By construction, the orthogonal projection of each cell
to the $x$-axis is an interval of length at most $\delta$.
(More precisely, we consider the polygon $P$ to be a closed set in the
plane. Subtracting the union of vertical lines $\mathcal{L}_V$ from $P$
results in a set of connected components; the closures of these
components are the simple polygons in $\mathcal{P}_V$).
The polygons in $\mathcal{P}_V$ are faces in the arrangement of
the lines in $\mathcal{L}_V$ and the edges of $P$; the planar dual of this
decomposition is a tree, whose nodes are the faces $\mathcal{P}_V$
and whose edges are dual to the vertical lines.
	
If the projection of polygon $Q\in \mathcal{P}_V$ onto the $x$-axis is
an interval of length $\delta$ (which means it extends from $x=i\delta$
to $x=(i+1)\delta$, for some integer $i$), we say that $Q$ is a
\emph{full-width} cell; otherwise the projection of $Q$ onto the $x$-axis
is of length less than $\delta$, and we say that $Q$ is a
\emph{narrow} cell. (It may be that $P$ itself is a narrow cell
if, e.g., $P$ does not intersect any of the vertical lines $L_V$.)

\begin{figure*}[h]
\centering
\includegraphics[scale=0.8]{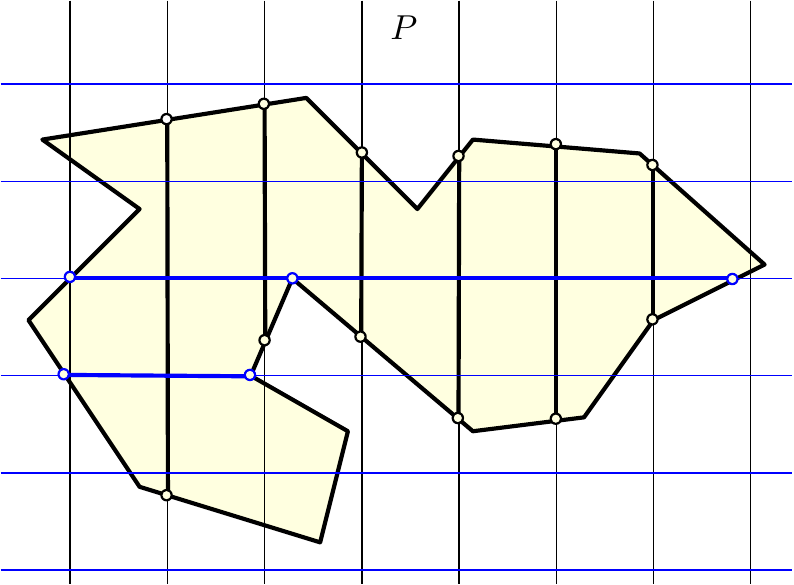}
\caption{$P$ is subdivided by a grid; the lasers are thick}
\label{fig:vertical}
\end{figure*}

The intersection of the lines in $\mathcal{L}_V$ with $P$ is a
set of vertical chords of $P$. Let $C_V$ be the set of these chords.
%Let $C_V$ be the chords of $P$ along the lines in $\mathcal{L}_V$.
While there is a chord $\ell\in C_V$ that lies on the boundary of
some narrow cell, remove $\ell$ from $C_V$ (thereby merging the cells
on the two sides of $\ell$ into one cell). As a result, all remaining
chords lie on the boundary between full-width cells.
Let $C_V'$ be the resulting set of chords,
and let $k_V=|C_V'|$ denote their cardinality.
%(Note that it could be that no vertical chords remain ($C_V'=\emptyset$) after this removal process; this happens if $P$ is intersected by only one vertical line from $L_V$ or if $P$ is not intersected by any vertical lines from $L_V$.)
%This removal of lasers (i.e., merging of faces in the partitioning of $P$) can lead to faces
%that are the union of many full-width cells (polygons of $\mathcal{P}_V$) that all lie within the same vertical strip (between two adjacent vertical lines of $L_V$) becoming connected, by narrow
%cells that may lie in the adjacent vertical strips.
%See the leftmost piece in Figure~\ref{fig:vertical}. However,
Since any two full-width cells of $\mathcal{P}_V$ that are in
adjacent vertical strips remain separated by a chord in $C_V'$,
the $x$-extent of each face in the new decomposition of $P$ is
at most $3\delta$. We summarize this below.
	
\begin{proposition}
The remaining $k_V$ chords $C_V'$, $k_V\geq 0$,  subdivide $P$ into a set $\mathcal{Q}$
of $k_V+1$ polygons, each of which intersects at most two lines in $\mathcal{L}_V$,
consequently its projection to the $x$-axis is an interval of length less than $3\delta$.
Further, the dual graph of this decomposition is a tree (with $k_V$ edges and $k_V+1$ nodes).
\end{proposition}

If $k_V=0$, then there is just one cell, $\mathcal{Q}=\{P\}$.
If $k_V\geq 1$, then each $Q\in \mathcal{Q}$ includes at least one full-width cell,
since the only lasers remaining are those separating one full-width
cell from an adjacent full-width cell sharing the laser.
	
Thus, the boundary of each $Q\in \mathcal{Q}$ includes at least two distinct
(simple) paths connecting a point on one line of $\mathcal{L}_V$ to a point on
an adjacent line of $\mathcal{L}_V$. Each of these paths has length at least
$\delta$. The endpoints of such a path are at distance at least $\delta$ away from each other.
In any laser cutting of $P$ into pieces of diameter at most $\delta$,
each such path contains a laser endpoint in its interior
or at both endpoints (if the path is a horizontal line segment).
In any case, each of these paths contains a laser endpoint in its
interior or at its left endpoint.
Thus overall, there must be at least $2|\mathcal{Q}|=2(k_V+1)$ endpoints
of lasers. This implies that $k^*\geq k_V+1$, where $k^*$ is the
minimum number of lasers in order to achieve pieces
of diameter at most~$\delta$. Therefore we conclude,

\begin{lemma}\label{lem:lb}
If $k_V\geq 1$, then $k^*\geq k_V+1$.
\end{lemma}
	
Now, we consider the set of horizontal lines,
$\mathcal{L}_H=\{y=j\delta: j\in \mathbb{Z}\}$,
and apply the above process to polygon $P$,
yielding horizontal chords $C_H$, and then a subset
$C_H'\subseteq C_H$ of chords after merging cells
(removing lasers that separate a \emph{full-height} cell
from an adjacent \emph{short} cell). The result is a
decomposition of $P$ into $k_H+1=|C_H'|+1$ pieces, each
having projection onto the $y$-axis of length less than $3\delta$.
Analogously to Lemma~\ref{lem:lb}, we get $k^*\geq k_H+1$ if $k_H\geq 1$.
	
If we now overlay the vertical chords $C_V'$ and the horizontal chords
$C_H'$, the resulting arrangement decomposes $P$ into pieces each of
which is a simple polygon having projections onto both the $x$- and
the $y$-axis of lengths less than $3\delta$; thus, the resulting pieces
each have diameter less than $3\delta\sqrt{2}$.
The total number of lasers is $k_V+k_H \leq 2(k^*-1)$.
	
\begin{theorem}\label{thm:bicriteria-diameter}
Let $P$ be a simple polygon with $n$ vertices,
and let $k^*$ be the minimum number of lasers that decompose $P$
into pieces each of diameter at most $\delta$ for a fixed $\delta>0$.
One can compute a set of at most $2(k^*-1)$ axis-aligned lasers
that decompose $P$ into pieces each of diameter at most $3\sqrt{2}\delta$
in time polynomial in $n$ and ${\rm diam}(P)/\delta$.
\end{theorem}

\subsection{$O(1)$-Approximation for \md in simple polygons}
\label{sec:k-laser-mindiameter}

In this section we consider the problem of minimizing the maximum diameter of a cell in the arrangement of $k$ lasers, for a given number \k. Our $O(1)$-approximation algorithm repeatedly decreases the $x$- and $y$- separation in the bi-criteria solution from Theorem~\ref{thm:bicriteria-diameter} until the number of placed lasers is about to jump over 2\k; then, the number of lasers is halved while increasing the diameter by a constant factor.

Specifically, let $\ell(\delta)$ denote the number of lasers used in the end of the bi-criteria algorithm with the $x$- and $y$-separation between consecutive vertical and horizontal lines being \d. Our algorithm to approximate the diameter achievable with \k lasers is as follows:

\begin{itemize}
	\item Initialize $\delta = {\rm diam}(P)$, and $\epsilon>0$
    \item While $\ell(\d)\le2k$, set $\d=\d/(1+\epsilon)$ and recompute $\ell(\d)$.
	\item Let $\delta_0$ be such that $\ell(\delta_0) \le 2k$ but $\ell(\delta_0/(1+\epsilon)) > 2k$.
	\item Let $C_V$ and $C_H$ be the $\ell(\delta_0) \le 2k$ vertical and horizontal lasers, resp., found by the bi-criteria algorithm.
	\item Partition $C_V$ into lasers along $x=i\delta_0$ for even $i$ and the rest (odd $i$); let $C_V'$ be a smallest part. Similarly, let $C_H$ be a smaller part when we partition $C_H$ into two subsets of lines where $y=i\delta_0$ is an even or odd multiple of $\delta_0$.
    \item Return the lasers in $C_V'\cup C_H'$.% (note that $|C_V'\cup C_H'|\leq k$).
\end{itemize}

\begin{theorem}\label{thm:approx-k-laser-diameter}
Let $P$ be a simple polygon with $n$ vertices, and let $\delta^*$ be the optimal diameter achievable with $k$ lasers.
For every $\epsilon>0$, one can compute a set of at most $k$ lasers that partition $P$ into pieces each of diameter at most $4\sqrt{2}(1+\epsilon) \d^*$ in time polynomial in $n$, ${\rm diam}(P)/\delta^*$, and $\epsilon$.
%\val{replaced $\log({\rm diam}(P)/\delta^*)$ with ${\rm diam}(P)/\delta^*$ -- we still have to use the running time from Theorem~\ref{thm:bicriteria-diameter}, right?} \csaba{Right.}
\end{theorem}

\begin{proof}
By Theorem~\ref{thm:bicriteria-diameter}, if $\d^*$ were smaller than $\d_0/(1+\epsilon)$, then $\ell(\d_0/(1+\epsilon))$ would have been at most $2k$, which is not the case (by the choice of $\d_0$), implying that $\d^*\ge\d_0/(1+\epsilon)$. Our algorithm starts with at most $2k$ lasers, produced by the bi-criteria solution, that decomposes $P$ into pieces that each intersects at most two consecutive lines
in $L_V=\{x=i\delta_0:i\in \mathbb{Z}\}$ and $L_H=\{y=i\delta_0:i\in \mathbb{Z}\}$, hence their $x$- and $y$-projections have length at most $3\d_0$. By removing at least half of the horizontal (resp., vertical) lasers, the number of lasers drops to \k or below, and the pieces on opposite sides of these lasers merge.
The removal of a laser along a line $x=i\delta_0$ creates a piece that can intersect only $x=j\delta_0$
for $j\in \{i-1,i,i+1\}$. Therefore, each resulting piece intersects at most 3 consecutive lines in $L_V$ and $L_H$, respectively, they each have $x$- and $y$-projection of length at most $4\delta_0$. Hence the diameter of the final pieces is at most $4\sqrt{2}\delta_0\leq 4\sqrt{2}(1+\epsilon) \d^*$.
\end{proof}

\section{Axis-Parallel Lasers}
\label{sec:orthogonal-laser}

In this section we study \mld and \mla under the constraint that all lasers must be axis-parallel (the edges of \P may have arbitrary orientations). The algorithms for both problems start with a ``window partitioning'' \P into ``(pseudo-)histograms'' of stabbing number at most three, and are then tuned to the specific measures to partition the histograms. We use a simple sweepline algorithm for the diameter, and a dynamic program for the area. The main result is:
\begin{theorem}\label{thm:minlaser-measure-axis-paralleli}
	Let \P be a simple polygon with $n$ vertices and let $k^*$ be the minimum number of axis-parallel lasers needed to subdivide $P$ into pieces of area (diameter) at most $1$.
	There is an algorithm that finds $O(k^*)$ axis-parallel lasers that subdivide \P into pieces of area (diameter) at most $1$ in time polynomial in $n$ and $\area(\P)$ ($\diam(\P)$).
\end{theorem}

\subsection{Reduction to Histograms}
\label{ssec:window}

A \emph{histogram} is a simple polygon bounded by an axis-parallel line segment, the \emph{base}, and an $x$- or $y$-monotone polygonal chain between the endpoints of the base.

The \emph{window partition} of a simple polygon was originally used for the design of data structures that support link distance queries~\cite{LinkDistance,Suri90}. In this section, we use the axis-parallel version, which partitions a simple polygon $P$ into histograms such that every axis-parallel chord of $P$ intersects at most 3 histograms. Window partitions for orthogonal polygons can be computed by a standard recursion~\cite{LinkDistance,Suri90}; we use a modified version where we recurse until the remaining subpolygons are below the size threshold $1$. This modification guarantees termination on all simple polygons (not only orthogonal polygons).

\begin{figure*}[!ht]
	\centering
	\includegraphics[scale=1]{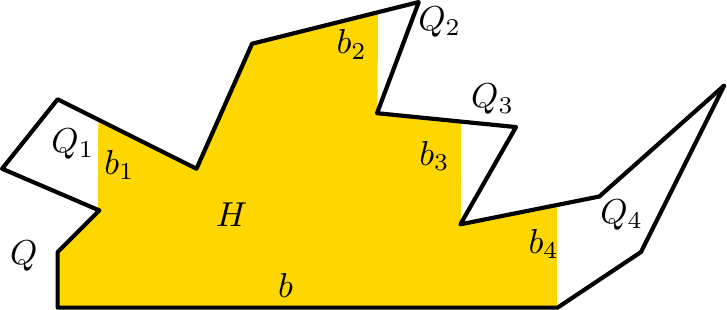}
	\caption{Window partition of a polygon $Q$ with a horizontal base $b$ into a maximal histogram $H$ (colored gold) and four polygons $Q_1$, $Q_2$, $Q_3$, and $Q_4$ (in white). If the sizes (areas/diameters) of $Q_1$, $Q_2$, $Q_3$, and $Q_4$ are each at most $1$, then $Q$ is a  pseudo-histogram.}
	\label{fig:window}
\end{figure*}

\smallskip\noindent\textbf{Window Partition Algorithm.}
Given a simple polygon \P, let $b$ be an axis-parallel chord of \P that subdivides \P into two simple polygons $P_1$ and $P_2$ with a common base $b$. Let $\mathcal{S}=\{(P_1,b),(P_2,b)\}$ be a set of tuples where each tuple has a polygon and its axis-parallel base, and let $\mathcal{H}=\emptyset$ be the set of histograms. While $\mathcal{S}$ contains a pair $(Q,b)$, where the size (e.g., the diameter)
of $Q$ is more than 1, do the following:
\benum
\item compute the maximal histogram\footnote{Without loss of generality, assume $b$ is horizontal. $H$ can be obtained by taking all points of $Q$ reachable through a vertical line from points on $b$.} $H$ of base $b$ in $Q$, and add $(H,b)$ to $\mathcal{H}$; See Figure~\ref{fig:window}.
\item update $\mathcal{S}$ by replacing $(Q,b)$ with the pairs $(Q_i,b_i)$, where the polygons $Q_i$ are the connected components of $Q\setminus H$, and $b_i$ is the boundary between $Q_i$ and $H$.
\eenum
Return $\mathcal{H}$ and $\mathcal{S}$.

\paragraph{Pseudo-histograms.}
%\csaba{We should change $\diam$ to ``measure'' that can be area, diameter, or inradius.}\val{said "e.g., diam" in the previous para}
Let $T_1$ and $T_2$ be the recursion trees of the algorithm, rooted at $P_1$ and $P_2$, respectively. Let $T=T_1\cup T_2$. Each node $v\in T$ corresponds to a polygon $Q_v\subset P$. Every nonleaf node $v\in T$ also corresponds to a histogram $H_v\subset Q_v$; it is possible that size$(H_v)\leq 1$ but size$(Q_v)>1$ (the size is $\area$ or $\diam$ based on the problem). For a leaf $v\in T$, we have either size$(Q_v)\leq 1$, or $H_v=Q_v$ and size$(H_v)>1$. The polygons $Q_v$ at leaf nodes and the histograms $H_v$ at nonleaf nodes jointly form a subdivision of $P$.

Every node $v\in T$ in the recursion tree corresponds to a polygon-base pair $(Q_v,b_v)$.
For any subset $U\subset V(T)$, where $V(T)$ is the set of vertices of $T$, the bases $\{b_u:u\in U\}$
decompose $P$ into simply connected cells. For every $u\in U$, there is a cell $P_u$ in the decomposition such that  $H_u\subset P_u\subset Q_u$. Since every axis-parallel chord of $P$ crosses at most 2 bases, it can intersect at most 3 polygons in such a decomposition. %\rathish{I don't understand the importance of a subset of vertices $U$.}\csaba{Each vertex is a cut-vertex in a tree; so a set of vertices decompose a tree into subtrees. Each vertex corresponds to a base (chord of \P), which decompose \P into cells.}

In a bottom-up traversal of $T$, we can find a subset $U\subset V(T)$ such that $\{b_u:u\in U\}$ decomposes
$P$ into polygons $P_u$, $u\in U$, such that size$(P_u)>1$ but the size of every component of $P_u\setminus H_u$ is at most $1$.
Each polygon $P_u$ consists of a histogram $H_u$ with base $b_u$, and subpolygons (\emph{pockets}) of size at most 1 attached to some edges of $H_u$ orthogonal to $b_u$. We call each such polygon $P_u$ a \emph{pseudo-histogram}. See Figure~\ref{fig:window}.

\subsection{$O(1)$-Approximation for \mld in Histograms}
\label{sec:min-orthogonal-laser-diameter}

We start with an $O(1)$-approximation for histograms, and then extend our algorithm to pseudo-histograms and simple polygons. Without loss of generality, we assume that the base is horizontal.

\begin{figure*}[!ht]
	\centering
	\includegraphics[scale=0.85]{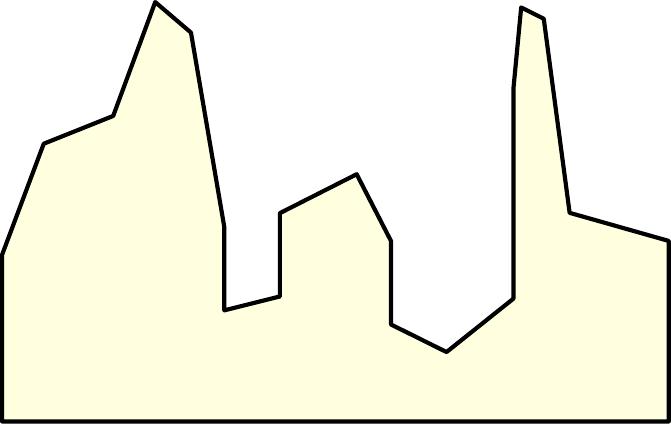}
	\includegraphics[scale=0.85]{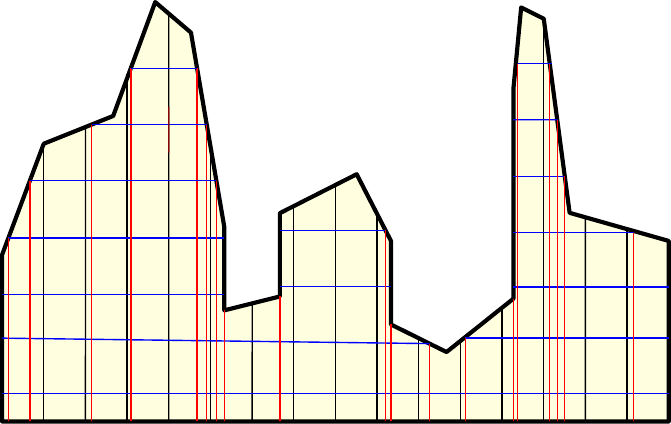}
	\caption{Left: A histogram polygon with a horizontal base. Right: lasers introduced in Phase~1 are shown in black.
    Horizontal (vertical) lasers introduced in Phase~2 are shown in blue (red). }
	\label{fig:histogram}
\end{figure*}

\begin{theorem}\label{thm:HistogramDiam}
	There is an algorithm that, given a histogram $P$ with $n$ vertices, computes an $O(1)$-approximation for the axis-parallel \mld problem in time polynomial in $n$ and $\diam(P)$.
\end{theorem}
\begin{proof}
We first describe the algorithm.

\smallskip\noindent\textbf{Algorithm.}
We are given a histogram $P$ with a horizontal base $ab$.
If $\diam(P)\leq 1$, halt. Otherwise, do the following:
\benum
\item  Subdivide $ab$ into $\lceil 2|ab|\rceil$ intervals which all, except possibly one, have length $1/2$
       and place vertical lasers on the boundary between consecutive intervals.
\item Sweep $P$ with a horizontal line $L$ top down, and maintain the
       set of cells formed by all lasers in step one and the line $L$.
       When the diameter of a cell $C$ above $L$ is precisely $1$,
       place a horizontal laser $pq$ along the bottom side of $C$,
       where $p,q\in \partial P$, and place two vertical lasers
       at $p$ and $q$, respectively.
\eenum

\smallskip\noindent\textbf{Analysis.}
Let $\opt$ denote the set of lasers in an optimal solution, and let $k^*=|\opt|$.
%and let $\opt_h$ and $\opt_v$ be the set of horizontal and vertical
%lasers in $\opt$.
Denote by $\alg$ the set of lasers computed by the algorithm;
let $\alg^1$ be the number of vertical lasers computed in Phase~1,
and let $\alg^2_h$ and $\alg^2_v$ be the set of horizontal and vertical
lasers computed in Phase~2. Clearly, $|\alg^2_v|\leq 2|\alg^2_h|$.
Therefore it is enough to prove that $|\alg^1|=O(k^*)$ and
$|\alg^2_h|=O(k^*)$.

First we show that $|\alg^1|=O(k^*)$. The vertical lasers in $\opt$
subdivide the base $ab$ into segments of length at most $1$.
Therefore, $k^*\geq \lfloor |ab|\rfloor$.
Combined with $k^*\geq 1$, this readily implies
$|\alg^1|=\lceil 2|ab|\rceil-1=O(k^*)$.

Next we prove $|\alg^2_h|\leq 2k^*$ using a charging scheme.
Specifically, we charge every laser in $\alg^2_h$ to a laser
in $\opt$ such that each laser in $\opt$ is charged at most twice.
Recall for each laser $pq\in\alg^2_h$ placed by the algorithm,
there is a cell $C=C_{pq}$ such that $\diam(C_{pq})=1$ and
$pq$ contains the bottom edge of $C_{pq}$. Since $\diam(C_{pq})=1$,
the cell $C_{pq}$ intersects some laser in $\opt$; we charge $pq$
to one of these lasers. Denote by $\opt_h(C_{pq})$ and $\opt_v(C_{pq})$,
resp., the horizontal and vertical lasers in $\opt$ that intersect $C_{pq}$.
The charging scheme is described by the following rules:
\begin{enumerate}\itemsep -2pt
\item[{\rm (a)}] If $\opt_h(C_{pq})\neq \emptyset$, then charge $pq$ to the
      lowest laser in $\opt_h(C_{pq})$;
\item[{\rm (b)}] else if $C_{pq}$ intersects $\partial P$, then
    charge $pq$ to a laser in $\opt_v(C_{pq})$ that
    is closest to $C_{pq}\cap \partial P$;
\item[{\rm (c)}] else if there is no horizontal laser in $\opt$ that lies above $pq$,
    then charge $pq$ to an arbitrary laser in $\opt_v(C_{pq})$;
\item[{\rm (d)}] else charge $pq$ to the lowest horizontal laser in $\opt$
        that lies above $pq$.
\end{enumerate}

It remains to prove that each laser in $\opt$ is charged at most once
for each the four rules. Since (a) and (d) charge horizontal lasers,
and (b) and (c) charge to vertical lasers in $\opt$, then
each laser in $\opt$ is charged by at most two of the rules.
In each case, we argue by contradiction.
Assume that a laser $\ell\in \opt$ is charged twice by one of the rules,
that is, there are two lasers  $pq,rs\in \alg^2_h$, that are charged to $\ell$ by
the same rule. The width of cells $C_{pq}$ and $C_{rs}$ is
at most $1/2$, because of the spacing of the vertical lasers in $\alg^1$.
Since $\diam(C_{pq})=\diam(C_{rs})=1$, they each have height
at least $\sqrt{3}/2$.
Without loss of generality, we may assume that the algorithm
chooses $pq$ before $rs$.

\smallskip\noindent\textbf{(a)}
In this case, $\ell$ is the lowest horizontal laser in $\opt$ that intersect
$C_{pq}$ and $C_{rs}$, respectively. Since $pq$ is above $rs$, laser $pq$ intersects
the interior of $C_{rs}$, contradicting the assumption that $C_{rs}$ is a cell
formed by the arrangement of all lasers in $\alg$.

\smallskip\noindent\textbf{(b)}
In this case, $\ell$ is a vertical laser that intersects both $C_{pq}$ and $C_{rs}$,
and also intersect $\partial P$. When the algorithm places
a horizontal laser at $pq$, it also places vertical lasers from $p$ and $q$ to
the base of $P$. These three lasers separate $\partial P$ from the portion of
$\ell$ below $pq$. This contradict the assumption that $C_{rs}$ is a cell
formed by the arrangement of all lasers in $\alg$.

\smallskip\noindent\textbf{(c)}
In this case, both $C_{pq}$ and $C_{rs}$ intersects a vertical laser
$\ell\in \opt$, and they both lie above the highest horizontal laser that
crosses $\ell$. Consequently, they both intersect the two highest cells,
say $C_{\rm left}$ and $C_{\rm right}$, on the two sides of $\ell$ in the
arrangement formed by $\opt$.
The combined height of $C_{pq}$ and $C_{rs}$ is at least $\sqrt{3}$.
Therefore, the height of $C_{\rm left}$ and $C_{\rm right}$ is at least
$\sqrt{3}>1$, contradicting the assumption that
$\diam(C_{\rm left})\leq 1$ and $\diam(C_{\rm right})\leq 1$.

\smallskip\noindent\textbf{(d)}
In this case, $\ell$ is the lowest horizontal laser in $\opt$
that lies above $C_{pq}$ and $C_{rs}$, respectively. Let $C_{\rm below}$
be the cell of the arrangement of $\opt$ that lies below $\ell$.
The combined height of $C_{pq}$ and $C_{rs}$ is at least $\sqrt{3}$.
Therefore, the height of $C_{\rm below}$ is at least $\sqrt{3}>1$,
contradicting the assumption that $\diam(C_{\rm below})\leq 1$.
\end{proof}

\paragraph{Adaptation to pseudo-histograms.}
In a laser cutting of $P$ into pieces of diameter at most 1, each pseudo-histogram $P_u$ intersects a laser, since $\diam(P_u)>1$. An adaptation of the algorithm in Section~\ref{sec:min-orthogonal-laser-diameter} can find an $O(1)$-approximation for \mld in each $P_u$.
As noted above, each laser intersect at most 3 pseudo-histograms, hence the union of lasers in the solutions for pseudo-histograms is an $O(1)$-approximations for $P$.

The sweep-line algorithm in Section~\ref{sec:min-orthogonal-laser-diameter} can easily be adapted to subdivide a pseudo-histogram $P_u$. Recall that $P_u$ consists of a histogram $H_u$ and pockets of diameter at most 1. We run steps~1 and~2 of the algorithm for the histogram $H_u$ with two minor changes in step~2: (1) we compute the \emph{critical} diameters with respect to $P_u$ (rather than $H_u$), and
(2) when the diameter of a cell $C$ above a chord $pq$ of $P_u$ is precisely 1, we place up to \emph{four} vertical lasers: at intersection points of $L$ with $\partial P_u$ the $\partial H_u$ (the vertical lasers through $pq\cap \partial H_u$ cut possible pockets that intersect $pq$). The analysis of the sweep line algorithm is analogous to Section~\ref{sec:min-orthogonal-laser-diameter}, and yields the following result.

%{\refthm{thm:SimPolygon-OrthogonalLasers}
%There is an algorithm that, given a simple polygon $P$ with $n$ vertices,
%computes an $O(1)$-approximation for the axis-parallel \mld problem in time polynomial in $n$ and $\diam(P)$.
%}
%

\begin{theorem}\label{thm:SimPolygon-OrthogonalLasers}There is an algorithm that, given a simple polygon $P$ with $n$ vertices,computes an $O(1)$-approximation for the axis-parallel \mld problem in time polynomial in $n$ and $\diam(P)$.\end{theorem}

\subsection{Discretization of a Histogram Polygon}
\label{app:discretization}

Consider a histogram polygon $P$ having $n$ vertices.
We assume that the vertices
are in general position, in the sense that no three vertices are collinear.

%In this section we show that there is a discrete set of candidate orthogonal chords such that lasers chosen from the candidate set yield an $O(1)$ approximation for minimizing the number of lasers, subject to a target measure bound $1$ on the obtained pieces. This is useful for finding an $O(1)$ approximation for \mla, to be described in the next subsection.

Let $V_P$ be the set of vertical chords in $P$ having top endpoint at a vertex of $P$;
such vertical chords yield the \emph{vertical decomposition} of $P$ into
vertical trapezoids (which are rectangles if $P$ is orthogonal).
Let $H_P$ be the set of horizontal chords in $P$ having at least one
of its endpoints at a vertex of $P$; such horizontal chords yield the
\emph{horizontal decomposition} of $P$ into horizontal trapezoids
(which are rectangles if $P$ is orthogonal).
The bottom side of a horizontal trapezoid is either the base of $P$ or
a chord in $H_P$; and the top side is a horizontal line segment (possibly of zero length)
that contains up to three chords in $H_P$ (since no three vertices of $P$
are collinear, and each vertex is incident to at most two horizontal chords of $P$).

\paragraph{\mla.}
We show that an $O(1)$-approximate solution for axis-parallel \ma on a histogram $P$
can be found among a discrete set of \emph{candidate} lasers.

%{\reflemma{lem:cl-area}
%For a histogram $P$, let $k^*$ be the minimum number of axis-parallel lasers that subdivide
%$P$ into pieces of area at most 1. We can compute a set $C$ of $O(n+\area(P))$ chords of $P$,
%such that $O(k^*)$ lasers from $C$ can subdivide $P$ into pieces of area at most 1.
%}

\begin{lemma}\label{lem:cl-area}
For a histogram $P$, let $k^*$ be the minimum number of axis-parallel lasers that subdivide
$P$ into pieces of area at most 1. We can find a set $C$ of $O(n+\area(P))$ chords of $P$,
such that $O(k^*)$ lasers from $C$ can subdivide $P$ into pieces of area at most 1.
\end{lemma}

\begin{proof}
We construct the set $C$ of candidate lasers as follows. We add all chords in $V_P\cup H_P$ (incident to vertices of $P$) into $C$.
For every vertical (horizontal) trapezoid $\tau$ in the decomposition of $V_P$ ($H_P$), let $R_{\tau}$ be the axis-parallel bounding box of $\tau$. If $\area(R_{\tau})>1$, then include $\lceil \area(R_{\tau})\rceil-1$ evenly spaced vertical (horizontal) chords into $C$, which subdivide $\tau$ into trapezoids of area at most $1$. We give an upper bound on the number of lasers in $C$.
Since each vertex of $P$ is incident to at most 4 axis-parallel chords, $|V_P\cup C_P|\leq O(n)$.
For every trapezoid $\tau$, we have $\area(\tau)\leq \area(R_\tau)\leq 2\,\area(\tau)$.
Since $\sum_{\tau} \area(\tau)=\area(P)$, where we sum over vertical (horizontal) trapezoids,
then $\sum_\tau \lceil \area(R_{\tau})\rceil-1 \leq \area(P)$.
Overall, $|C|\leq O(n+\area(P))$, as required.
	
Let $\opt$ be a set of lasers in an optimal solution for the axis-parallel \mla problem, where $k^*=|\opt|$. We choose a subset $L\subset C$ of $O(k^*)$ lasers that subdivide $P$ into pieces of area at most 1. We may assume that $\area(P)>1$, hence $k^*\geq 1$ (otherwise we could choose $L=\emptyset$).
	
For each vertical laser $\ell\in \opt$, we add the nearest vertical chords in $C$ on the left and right side of $\ell$, resp., into $L$. Similarly, for each horizontal laser $\ell\in \opt$, we add the nearest top and bottom horizontal chords of $\ell$ from $H_P \cup H_1$ into $L$.
There is at most one nearest horizontal chord in $C$ below $\ell$, and at most three nearest horizontal chords above $\ell$,
as no three vertices of $P$ are collinear. Hence, we place $O(1)$ lasers for each laser $\ell\in \opt$.
	
The lasers in $L$ subdivide $P$ into new cells; let $\sigma$ be one of them. We claim that $\area(\sigma)\leq 1$.
Suppose, to the contrary, that $\area(\sigma)>1$. Then some laser $\ell\in \opt$ intersects the interior of $\sigma$.
If $\ell$ is vertical (horizontal), then $\ell$ is in a vertical (horizontal) trapezoid $\tau$ with $\area(\tau)\leq 1$
such that $\tau$ is bounded by vertical (horizontal) lasers in $L$. Consequently, $\sigma\subseteq \tau$,
hence $\area(\sigma)\leq \area(\tau)\leq 1$. This contradicts our assumption, and proves the claim.
\end{proof}

\paragraph{\mld.}
We show that an $O(1)$-approximate solution for axis-parallel \md on a
histogram $P$ can be found among a discrete set of \emph{candidate} lasers.

%{\reflemma{lem:cl-diam}
%For a histogram $P$, let $k^*$ be the minimum number of axis-parallel lasers that subdivide
%$P$ into pieces of diameter at most 1. We can compute a set $C$ of $O(n+\text{per}(P))$ chords of $P$,
%such that $O(k^*)$ lasers from $C$ can subdivide $P$ into pieces of diameter at most 1.
%}
\begin{lemma}\label{lem:cl-diam}
For a histogram $P$, let $k^*$ be the minimum number of axis-parallel lasers that subdivide
$P$ into pieces of diameter at most 1. We can find a set $C$ of $O(n+\text{per}(P))$ chords of $P$,
such that $O(k^*)$ lasers from $C$ can subdivide $P$ into pieces of diameter at most 1.
\end{lemma}
\begin{proof}
Let $\mathcal{L}_V=\{x=i\cdot (1/\sqrt{2}):i\in \mathbb{Z}\}$ be an infinite set of
vertical lines. Let $V_1$ be the set of vertical chords of $P$ that lie on lines
in $\mathcal{L}_V$. As $P$ is $x$-monotone, $P$ has at most one chord in any line
in $\mathcal{L}_V$. Similarly, let $H_1$ be the set of horizontal chords of $P$ that
lie on lines in $\mathcal{L}_H=\{y=i\cdot (1/\sqrt{2}):i\in \mathbb{Z}\}$;
any horizontal line can contain up to $\Omega(n)$ chords of $P$.
Note that $|V_1|\leq O(\diam(P))$, and $|H_1|\leq O(\text{per}(P))$.
Letting $C=(H_P\cup V_P)\cup (H_1\cup V_1)$, we have $|C|\leq O(n+\text{per}(P))$.

Let $\opt$ be a set of lasers in an optimal solution for the axis-parallel \mld problem, where $k^*=|\opt|$.
We choose a subset $L\subset C$ of $O(k^*)$ lasers that subdivide $P$ into pieces of diameter at most 1.
We may assume that $\diam(P)>1$ hence $k^*\geq 1$ (otherwise we could choose $L=\emptyset$).

Denote by $w$ the length of the base edge of $P$. Then $|V_1|\leq \ceil{w/(1/\sqrt{2})} = \ceil{\sqrt{2}w}$.
Since the vertical lasers in $\opt$ subdivide the base into intervals of length at most 1, we have
$k^*\geq \ceil{w}$. Combined with $k^*\geq 1$, this implies $|V_1|\leq O(k^*)$.
We include all vertical lasers in $V_1$ to $L$.

For each vertical laser in $\ell\in \opt$, we add $O(1)$ lasers from $V_P\cup H_P\cup H_1$ into $L$ as follows.
The laser $\ell$ lies in some vertical trapezoid, $\tau$, in the subdivision by $V_P$.
We add both lasers of $V_P$ on the boundary of $\tau$ into $L$.
We also add the highest chord from $H_P\cup H_1$ that intersects $\ell$ into $L$.

For each horizontal laser $\ell\in \opt$, we add $O(1)$ lasers from $V_P\cup V_1\cup H_1$
into $L$ as follows. Let $R_{\ell}\subset P$ be a maximal axis-parallel rectangle whose
top side is $\ell$ and the interior of $R_\ell$ is disjoint from horizontal lasers in $\opt$
(i.e., the bottom side of $R_\ell$ is either the base of $P$ or another horizontal
laser in $\opt$). We add all lasers in $H_1$ that intersect $R_{\ell}$ into $L$.
Note that the height of $R_{\ell}$ is at most $1$ (otherwise $\opt$
would contain a cell of height more than $1$). Thus $R_\ell$ intersects at most
two lasers from $H_1$ (that we add in $L$).
Furthermore, $\ell$ lies in some trapezoid, $\tau$, in the subdivision formed by $H_P$.
We add all $O(1)$ lasers of $H_P$ on the boundary of $\tau$ to $L$;
and if any chord in $H_1$ intersects $\tau$ above $\ell$,
then we add the lowest such chord to $L$.

The lasers in $L$ subdivide $P$ into new cells; let $\sigma$ be one of them.
We claim that $\diam(\sigma)\leq 1$.
Since $V_1\subset L$, the $x$-extent of $\sigma$ is at most $1/\sqrt{2}$.
If a topmost edge of $\sigma$ is in a horizontal laser $\ell\in L$, then
by construction $L$ contains another horizontal lasers at distance at most
$1/\sqrt{2}$ below $\ell$. Consequently, the $y$-extent of $\sigma$ is also
at most $1/\sqrt{2}$, and $\diam(\sigma)\leq 1$.

It remains to consider the case that $\sigma$ is bounded above the boundary of $P$.
Suppose, to the contrary, that $\diam(\sigma)>1$, hence its $y$-extent is more than $1/\sqrt{2}$.
Then some laser $\ell\in \opt$ intersects $\sigma$.
If $\ell$ is vertical, then the top endpoint of $\ell$ is in $\sigma$.
By construction, $L$ contains the highest chord from $H_P\cup H_1$ that intersects $\ell$.
Thus the $y$-extent of $\sigma$ is at most $1/\sqrt{2}$.
If $\ell$ is horizontal, lying in some horizontal trapezoid $\tau$, then $L$ contains
all lasers in $H_P$ along the top edge of $\tau$. If the top edge of $\tau$ is an edge of $P$,
then $L$ also contains the lowest a chord in $H_1$ that intersects $\tau$ above $\ell$.
This again implies that the $y$-extent of $\sigma$ is at most $1/\sqrt{2}$.
In both cases, we have shown that he $y$-extent of $\sigma$ is at most $1/\sqrt{2}$,
hence $\diam(\sigma)\leq 1$, as claimed.
\end{proof}

\subsection{$O(1)$-Approximation for \mla}
\label{app:min-orthogonal-laser-DP}

We now consider the problem of \mla, with axis-parallel lasers chosen
from a discrete set to achieve pieces of area at most $1$.
The $O(1)$-approximation algorithm is based on the window partition method
described earlier, allowing us to reduce to the case of subdividing a
pseudo-histogram $P$, for which we give an exact dynamic programming
algorithm.

Let $V$ (resp., $H$) be the discrete set of vertical (resp., horizontal) candidate lasers.
For $h\in H$, let $P_h \subseteq P$ be the sub-pseudo-histogram of $P$ that is
above $h$, with base $h$.
% We call $P_h$ a subhistogram with base $h$.
%A subproblem $\sigma$ corresponds to a subhistogram of $P$, with a base that is either the base of $P$ or one of the candidate horizontal lasers in $H$.
We note that in an optimal solution for $P$ or for $P_h$, we need not consider vertical lasers from $V$ other than those
that meet the base of the pseudo-histogram; any vertical laser within a ``pocket'' of the pseudo-histogram can be replaced by
the laser that defines the lid of the pocket, since each pocket is of area at most 1.

A subproblem in the dynamic program is specified as a tuple,
\[\sigma=(h_{\rm base}, k_v, v_{\rm left}, v_{\rm right}, h_{\rm over}, W_{\rm max}),\]
which includes the following data:

\begin{description}

\item[$h_{\rm base}$] is either the base of $P$ or one of the candidate
horizontal lasers, $H$.

\item[$k_v$] is the number of vertical lasers extending through the
base, $h_{\rm base}$, of the subproblem.

\item[$v_{\rm left}$ and $v_{\rm right}$] are the leftmost and rightmost vertical lasers
extending through the base, $h_{\rm base}$, of the sub-pseudo-histogram
(the positions of other vertical lasers are not specified by $\sigma$).
Possibly, $v_{\rm left}=v_{\rm right}$, if $k_v=1$. If $k_v=0$,
we set $v_{\rm left}=v_{\rm right}=NIL$.

\item[$h_{\rm over}$] is an ``overhanging'' horizontal laser, which is a constraint inherited from a neighboring subproblem; it is either $NIL$ (no constraint) or is a horizontal laser of $H$ that crosses
$v_{\rm left}$, and is above the base $h_{\rm base}$.
Further, if $h_{\rm over}\neq NIL$, then only $v_{\rm left}$ can cross
$h_{\rm over}$ (out of the $k_v$ vertical lasers extending through the
base, $h_{\rm base}$), and no horizontal laser can cross $v_{\rm
  left}$ above $h_{\rm base}$ and below $h_{\rm over}$.

\item[$W_{\rm max}$] is the maximum allowed spacing between consecutive vertical lasers
in the subproblem; it suffices to consider values $W$ that are determined by pairs of candidate vertical lasers.
\end{description}

If $h_{\rm over}=NIL$, the sub-pseudo-histogram, $P_\sigma$, corresponding to
$\sigma$ is $P_\sigma$; otherwise, the sub-pseudo-histogram
$P_\sigma$ is the subset of $P_\sigma$ that is to the right of
$v_{\rm left}$ and below $h_{\rm over}$.

Our goal is to compute the function $f(\sigma)$, equal to the minimum
number of horizontal lasers that (together with a suitable set
of $k_v$ vertical lasers) can partition the sub-pseudo-histogram
corresponding to $\sigma$ into pieces of area at most $1$,
subject to the parameters of the subproblem.
Note that, by our choice of candidate horizontal lasers (cf.~Lemma~\ref{lem:cl-area}),
the candidate set $H$ is sufficient to guarantee that it is always possible to
achieve pieces of area at most $1$, even without vertical lasers.
%
% We define $f(\sigma)=+\infty$ if it is not possible to achieve piece area at most $1$ for subproblem $\sigma$.

If there exists a set of vertical lasers, $v_1,v_2,\ldots,v_{k_v}$,
between $v_{\rm left}=v_1$ and $v_{\rm right}=v_{right}$, such that the areas
of the resulting subpieces of $P_\sigma$, using only these vertical
lasers, are each at most $1$, then $f(\sigma)=0$, since no
horizontal lasers are needed to achieve the desired area threshold of
$1$.
It is easy to decide if this is the case:
(i) if $h_{\rm over}=NIL$, then we must have that the area of the portion
of $P_\sigma=P_{h_{\rm base}}$ that is left of $v_{\rm left}$ is at most
$1$, the area of the portion of $P_\sigma$ to the right of
$v_{\rm right}$ is at most $1$, and there exists a set of
intermediate vertical cuts, $v_2,\ldots,v_{k_v-1}$, ordered from left
to right between $v_{\rm left}$ and $v_{\rm right}$, with each piece of
$P_\sigma$ between $v_i$ and $v_{i+1}$ having area at most
$1$. (A simple greedy strategy allows this to be checked, placing
lasers from left to right in order to make each piece be as wide as
possible, while having area at most $1$.)
(ii) if $h_{\rm over}\neq NIL$, then we proceed similarly, but now
within the pseudo-histogram polygon $P_\sigma$, which is bounded on the left
by $v_{\rm left}$, and lies below $h_{\rm over}$.

If $k_v\leq 2$, then we compute the fewest horizontal lasers (from $H$) to
meet the area bound $1$ by sweeping $P_\sigma$, from top to bottom,
inserting horizontal lasers (from the discrete candidate set) only as
needed to achieve the area bound.
(Recall that, by our choice of discrete candidate lasers, the area bound can always
be achieved.)

\begin{figure*}[!ht]
	\centering
	\includegraphics[width=\textwidth]{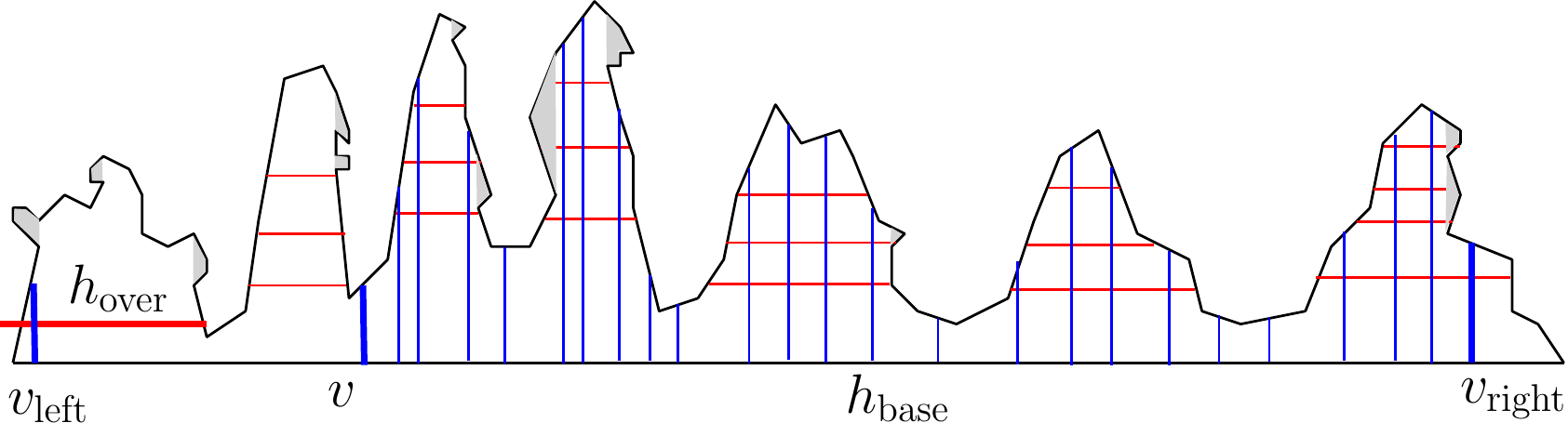}\\
	\includegraphics[width=\textwidth]{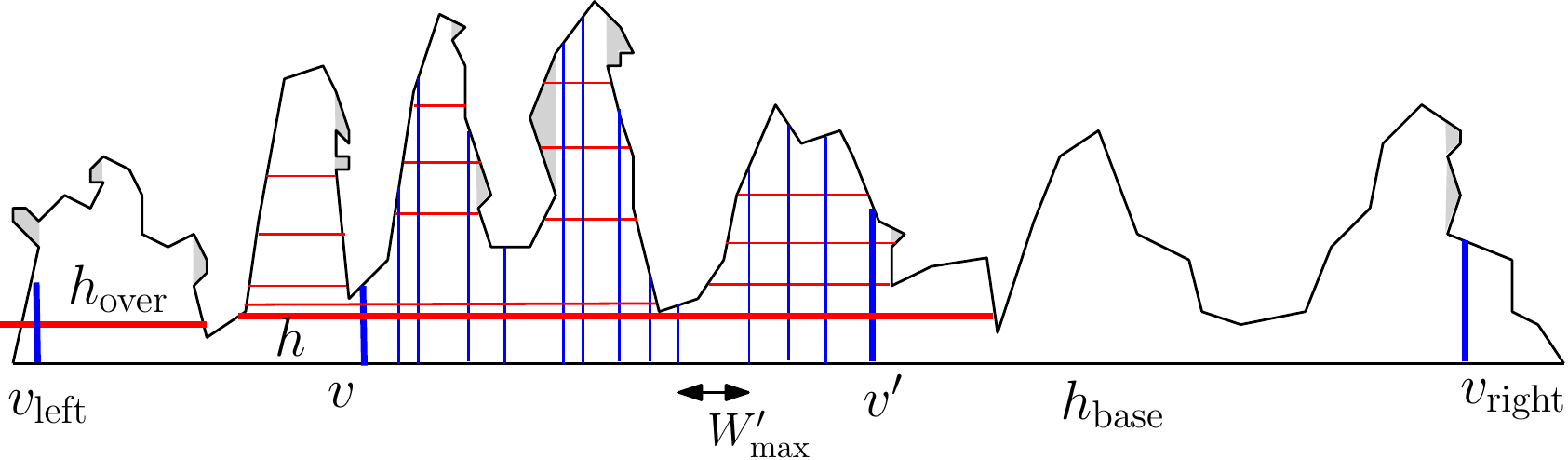}
	\caption{Top: Case (1): $v$ is not crossed by a horizontal laser.
          Bottom: Case (2) $v$ is crossed by a horizontal laser; $h$ is the lowest such laser. The pockets are shown in gray shading.}
	\label{fig:subproblems}
\end{figure*}

If $k_v\geq 3$, then we optimize over the choice of the leftmost
vertical laser cut, $v$, within $P_\sigma$, considering two possibilities:
\begin{description}
\item[(1)] In the optimal solution, $v$ is not crossed within $P_\sigma$ by a
horizontal laser.

In this case,
$$f(\sigma) = \min_{v} \{f(\sigma') + H(v_{\rm left},v,h_{\rm over})\},$$
with $H(v_{left},v,h_{over})$ equal to the minimum number of horizontal lasers
required in the sub-pseudo-histogram determined by $(v_{\rm left},v,h_{\rm over})$,
with $\sigma'=(h_{\rm base}, k'_v=k_v-1, v'_{\rm left}=v, v_{\rm right}, h'_{\rm over}=NIL, W_{\rm max})$, and
with the minimization being taken over $v\in V$ satisfying (i) $v$ lies strictly between
(in $x$-coordinate) $v_{\rm left}$ and $v_{\rm right}$,
(ii) $v$ is at distance at most $W_{\max}$ to the right of $v_{\rm left}$, and
(iii) $v$ does not cross $h_{\rm over}$ (if $h_{\rm over}\neq NIL$).

\item[(2)] In the optimal solution, $v$ is crossed by at least one horizontal laser.
Let $h$ be the lowest such horizontal laser crossing $v$,
let $v'$ (possibly $v'=v$) be the rightmost vertical laser of OPT crossing $h$,
let $k'\in\{1,2,\ldots ,k_v\}$ be the number of vertical lasers crossing $h$ in OPT,
and let $W'_{max}$ be the maximum spacing between consecutive vertical
lasers crossing $h$ in OPT.

Necessarily, the area of the piece of $P_\sigma$ that is left of $v$
and below $h_{\rm over}$ and below $h$ must be at most $1$.
(Note that the inherited overhang constraint implies the
constraint that no additional horizontal
  laser can cross $v_{\rm left}$ between $h_{\rm base}$ and $h_{\rm over}$.)

We get new subproblems
$\sigma_1=(h, k', v, v', NIL, W'_{\rm max})$, and
$\sigma_2=(h_{\rm base}, k_v-k', v', v_{\rm right}, h, W_{\rm max})$.

In this case, the recursive optimization is given by
$$f(\sigma) = \min_{v,v',h,k'} \{ 1+ f(\sigma_1) + f(\sigma_2)\}.$$
\end{description}
This concludes the description of the dynamic program.

While we have specified the algorithm for the measure of area, with
slight modifications, the algorithm also applies to the measure of
diameter, allowing us to solve \mld in pseudo-histograms
(at a much higher polynomial time bound than stated in Theorem~\ref{thm:HistogramDiam}.
%%%%%%%%%%%%%%%

\section{Diameter Measure in Polygons with Holes and Axis-Parallel Lasers}
\label{sec:polygon-holes}

\subsection{Bi-Criteria Approximation for Diameter}
\label{sec:bicriteria-holes}
In this section we give a bi-criteria approximation for the diameter version
in a polygon with holes when lasers are constrained to be axis-parallel.
The approach is similar to the algorithm for simple polygons
and lasers of arbitrary orientations (cf.~Section~\ref{sec:minlaser-diameter}) in that both use grid lines, but they differ significantly to handle holes in a polygon when the lasers are axis-parallel. Particularly, in simple polygons we place lasers along grid lines, while in polygons with holes the grid lines just divide the problem into sub-problems.

\paragraph{Lasers in vertical strips.}
Consider the infinite set of equally spaced vertical lines $\mathcal{L}_V=\{x=i\delta: i\in \mathbb{Z}\}$, for some $\d>0$.
The lines subdivide \P into a set $\mathcal{P}_V$ of polygons (possibly with holes),
that we call \emph{strips}. (Unlike Section~\ref{sec:minlaser-diameter}, we do not place
lasers along the lines in $\mathcal{L}_V$; we use the strips for a divide-and-conquer strategy.)
The projection of any strip on the $x$-axis has length at most $\delta$;
we say that a strip is \emph{full-width} if its projection has length exactly $\delta$.
Let $\mathcal{F}_V\subset \mathcal{P}_V$ denote the set of full-width strips,
and let $F\in\mathcal F_V$ be a full-width strip.

The leftmost (resp., rightmost) points of $F$ lie on a line $L=\{x=i\delta\}$ (resp., $R=\{x=(i+1)\delta\}$) for some $i\in \mathbb{Z}$ (see Fig.~\ref{fig:F}). Consequently, the outer boundary of $F$ contains two simple paths between $L$ and $R$; we denote them by $T$ (top) and $B$ (bottom).
\begin{figure}\centering\includegraphics{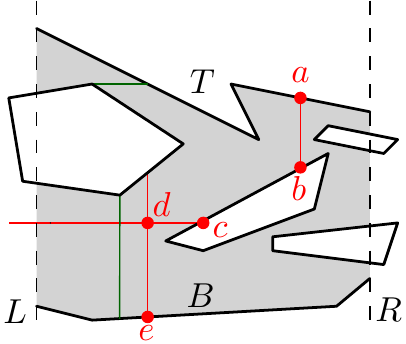}\caption{$F$ is shaded, the holes are white. The $L$-$R$ separating path $\gamma=abcde$ (vertices marked with red disks) alternates between holes and lasers (red) in the interior of \P; $a_F(\g)=2$ as there are two links $ab$ and $de$ in path \g whose extensions are fully contained in $F$. $ab$ and $cde$ are the maximal rectilinear subpaths of \g through the free space. The minimum-link path $\gamma(F)$ (darkgreen) also alternates between holes and free space.}\label{fig:F}\end{figure}

Since the distance between $L$ and $R$ is $\delta$, in any laser cutting of $P$ into cells of diameter at most $\d$, there exists a $T\textrm-B$ path $\gamma\subset F$ along the boundaries of cells that separates $L$ and $R$. Since $\gamma$ is disjoint from the interior of the cells, it must follow lasers in the interior of $P$. %; in addition, it can also follow the boundary of $P$ arbitrarily.
We may assume, w.l.o.g., that \g follows any laser at most once; otherwise we could shortcut \g along the laser. Since the lasers are axis-aligned, \g is an alternating sequence of subpaths that are either in $\partial P$ or rectilinear paths through the interior of $F$; we call any such $T\textrm-B$ path an \e{alternating} path.

An axis-aligned segment $s$, fully contained in $F$, is \e{associated} with $F$ if it remains fully contained in $F$ after it is maximally extended within \P (i.e., if both endpoints of the supporting chord of $P$ are on the boundary of $F$). For example, any \e{vertical} segment $s\subset F$ is associated with $F$ (because $T$ and $B$ belong to the boundary of $F$). Let $a_F(\g)$ be the number of associated links of $\g$ (i.e., the number of edges whose supporting chords are fully contained in $F$). Let $|\g|$ denote the total number of the (axis-aligned) edges in $\g$. A key observation is the following.

\begin{lemma}\label{lem:vertical-piece}
$|\g|\le3\,a_F(\g)$.
\end{lemma}
\begin{proof}
Let $\pi$ be a (maximal) rectilinear subpath of \g through the free space, i.e., a part of \g whose endpoints lie on the boundary of \P. If $\pi$ is a single horizontal link, then the link is associated  with $F$ (because if any of its two ednpoints is outside $F$, then \g protrudes through $L$ or $R$, not separating them). Otherwise (i.e., if $\pi$ contains vertical links), the number of the vertical links is at least 1/3 of the total number of links in $\pi$. The lemma follows by summation over all subpaths of $\gamma$.
\end{proof}

Our algorithm computes an alternating path $\g(F)$ with the minimum number of links and places one laser along every link of $\gamma(F)$ (the horizontal lasers may extend beyond $F$).
%
% \csaba{I've commented out the algorithm using Hanan grids.}
%To find $\g(F)$, we build the \e{Hanan grid} \cite{hanan1966steiner} of $F$, which is the arrangement of maximal horizontal and vertical segments though vertices of $F$; it is well known that the grid is rich enough to contain a minimum-link path.
%\csaba{Do we have a reference for this? It seems to be true, but the number of alternations (between rectilinear vs. $\partial P$ subpaths) might increase when we ``snap'' a path to the Hanan grid...}
%\val{put below another version of the paragraph about computing $\g(F)$, bringing back wave propagation (albeit with different refs -- please correct if needed); please choose or comment further}
%Since we are looking for the minimum-link \e{alternating} path, we ``contract'' the grid by connecting all vertices touching the same hole with 0-weight edges (any hole, essentially, provides ``shortcuts'' in the grid, identifying all vertices incident to the hole).
%
To find $\g(F)$, we can build the \e{critical graph} of $F$, whose vertices are $T$, $B$, and components of $\partial P$ within the strip $F$ (including holes in the strip), and in which the weight of the edge between two vertices is the axis-parallel link distance between them.
% (the terms ``critical graph,'' ``top,'' ``bottom'' are borrowed from studies of geometric flows and cuts \cite{arkin2010maximum,eriksson2014optimal,mitchell1990maximum,gewali1990path}).
% \csaba{Thee are quite descriptive terms, I think there is no need to explain them.}
The weight of an edge between vertices $i$ and $j$ can be found by in polynomial time by standard wave propagation techniques~\cite{mitchell2014minimum,das1991geometric}, i.e., by successively labeling the areas reachable with \k links from $i$ for increasing \k, until $j$ is hit by the wave. After the critical graph is built, $\g(F)$ is found as the shortest $T\textrm-B$ path in the graph.
%(We note that computing the complete critical graph may be an overkill, but we focus only on polynomiality.)

By minimality of $\g(F)$, the number links $|\g(F)|$ in it (and hence the number of lasers we place) is at most $|\g(F)|$. Let $k_V=\sum_{F\in \mathcal{F}_V} |\gamma(F)|$ be the total number of lasers placed in all full-width strips in $\mathcal{F}_V$, and let $k^*$ be the minimum number of axis-parallel lasers in a laser cutting of \P into cells of diameter at most $\delta$.
An immediate consequence of \lemref{vertical-piece} is the following.
\begin{corollary}\label{cor:vertical-piece}
$k_V\leq 3k^*$.
\end{corollary}
\begin{proof}
As the links of $\g(F)$ follow lasers, at least $a_F(\g)$ lasers are fully contained in $F$.
\end{proof}

The $k_V$ lasers placed in full-width strips subdivide $P$ into polygonal pieces; let $Q$ be one such piece.% that we denote by $\mathcal{Q}_V$.
\begin{lemma}\label{lem:2delta}
The length of the $x$-projection of $Q$ on the $x$-axis is at most $2\delta$.
\end{lemma}
\begin{proof}
We prove that $Q$ intersects at most one line in $\mathcal{L}_V$. %and is thus contained in the union of at most two vertical strips.
Suppose, to the contrary, that $Q$ intersects two consecutive lines $\ell_1:x=i\delta$ and $\ell_2=x=(i+1)\delta$.
Let $\lambda$ be a shortest path in $Q$ between points in $Q\cap \ell_1$ and
$Q\cap \ell_2$, respectively. By minimality, $\lambda$ lies in the strip
between $\ell_1$ and $\ell_2$. Consequently, $\lambda$ is contained in
some full-width strip $F\subset \mathcal{F}_V$. However, the path $\gamma(F)$
intersects every path in $F$ between $F\cap \ell_1$ and $F\cap \ell_2$;
in particular, it intersects $\lambda$. Since we have placed a laser
along every segment of $\gamma(F)$ in the interior of $P$,
$\lambda$ intersects a laser, contradicting the assumption that
$\lambda\subset Q$.
\end{proof}

\paragraph{Lasers in horizontal strips.}
Similarly, we consider the set of horizontal lines $\mathcal{L}_H=\{y=j\delta: j\in \mathbb{Z}\}$
and apply the above process to $P$, yielding horizontal chords $C_H$
that subdivide the polygon into horizontal strips (polygons, possibly with holes). We again work only with \emph{full-height} strips,
whose boundary intersect two consecutive lines in $\mathcal{L}_H$.
In each full-height strip, we find a minimum-interior-link rectilinear path that separates the boundary points
along the two lines in $\mathcal{L}_H$, and place lasers along the links of the path.
Let $k_H$ be the number of lasers over all full-height strips.
%
%An analogue of Corollary~\ref{cor:vertical-piece} shows that $k_H\leq 3k^*$,   By an analogue of Lemma~\ref{lem:2delta}, these lasers subdivide $\P$ into pieces, each of which has $y$-projection of length at most $2\d$.

\paragraph{Putting everything together.}
We overlay the $k_V$ lasers in full-width strips with the $k_H$ lasers in full-height strips.
The resulting arrangement partitions $P$ into polygonal pieces (possibly with holes).
The $x$- and $y$-projection of each piece has length at most
$2\delta$ by Lemma~\ref{lem:2delta}; thus, each piece has diameter less than $2\delta\sqrt{2}$.
By Corollary~\ref{cor:vertical-piece}, the total number of lasers used in the
arrangement is $k_V+k_H \leq 6k^*$. We obtain the following theorem.

\begin{theorem}\label{thm:bicriteria-diam-holes}
	Let $P$ be a polygon with holes of diameter ${\rm diam}(P)$ having $n$ vertices,
	and let $k^*$ be the minimum number of laser cuts that partition $P$
	into pieces each of diameter at most $\delta$ for a fixed $\delta>0$. In time
	polynomial in $n$ and ${\rm diam}(P)/\delta$, one can compute a set of
	at most $6k^*$ lasers that subdivide $P$ into pieces each of diameter
	at most $2^{3/2}\delta$.
\end{theorem}

\subsection{$O(1)$-Approximation to \md}\label{sec:constant-holes}

Similarly to Section~\ref{sec:k-laser-mindiameter}, we can use the bi-criteria algorithm to
derive a constant-factor approximation for minimizing the maximum diameter of a cell in the arrangement of a given number \k of axis-parallel lasers. Our $O(1)$-approximation algorithm repeatedly decreases the $x$- and $y$- separation in the bi-criteria solution from Theorem~\ref{thm:bicriteria-diam-holes} until the number of placed lasers is about to jump over 6\k; then, the number of lasers is decreased by a factor of $6$ while increasing the diameter by a constant factor.

Specifically, let $\ell(\delta)$ denote the number of lasers used in the end of the bi-criteria algorithm with the $x$- and $y$-separation between consecutive vertical and horizontal lines being \d. Our algorithm to approximate the diameter achievable with \k lasers is as follows:

\begin{itemize}
	\item Initialize $\delta = {\rm diam}(P)$, and let $\epsilon>0$.
    \item While $\ell(\d)\le 6k$, set $\d=\d/(1+\epsilon)$ and recompute $\ell(\d)$.
	\item Let $\delta_0$ be such that $\ell(\delta_0) \le 6k$ but $\ell(\delta_0/(1+\epsilon) > 6k$.
	\item Let $\mathcal{F}_V$ and $\mathcal{F}_H$ be the full-width and full-height strips, resp.,
 used in the bi-criteria algorithm to place the $\ell(\delta_0)$ lasers.
	\item Partition $\mathcal{F}_V$ into 6 subsets: the set of strips whose left boundary is in a line $x=i\delta_0$, where $i\equiv a\mod 6$ for $a=\{0,1,2,3,4,5\}$. Let $\mathcal{F}_V'\subset \mathcal{F}_V$ be a subset of strips that uses the minimum number of chords for their minimum-link paths.
	\item Similarly, partition $\mathcal{F}_H$ into 6 subsets of strips based on the residue class of $j\mod 6$, where the top side of the strip is in $y=j\delta_0$. Let $\mathcal{F}_H'\subset \mathcal{F}_H$ be a subset that uses the minimum number of chords for their minimum-link paths.
    \item Return the lasers used in minimum-link paths in the strips of $\mathcal{F}_V'$ and $\mathcal{F}_H'$.
\end{itemize}

\begin{theorem}\label{thm:approx-k-laser-diam-hole}%Let $P$ be a simple polygon with $n$ vertices, and
	Let $\delta^*$ be the minimum diameter achievable with $k$ axis-parallel lasers.
	For every $\epsilon>0$, one can compute a set of at most $k$ axis-parallel lasers that partition $P$ into pieces each of diameter at most $12\sqrt{2}(1+\epsilon)\d^*$ in time polynomial in $n$, ${\rm diam}(P)/\delta^*$, and $\epsilon$.
\end{theorem}

\begin{proof}
By Theorem~\ref{thm:bicriteria-diam-holes}, if $\d^*$ were smaller than $\d_0/(1+\epsilon)$, then $\ell(\d_0/(1+\epsilon))$ would have been at most $6k$, which is not the case (by the choice of $\d_0$), implying that $\d^*\ge\d_0/(1+\epsilon)$. Our algorithm starts with at most $6k$ lasers, produced by the bi-criteria solution, that decomposes $P$ into strips that each intersect at most one line in $\mathcal{L}_V=\{x=i\delta_0:i\in \mathbb{Z}\}$ and in $\mathcal{L}_H=\{y=i\delta_0:i\in \mathbb{Z}\}$, respectively; hence their $x$- and $y$-projections have length at most $2\d_0$. By removing at least $\frac56$ of the horizontal (resp.\ vertical) lasers, the number of lasers drops to \k or below, and the cells on opposite sides of these lasers merge.
However, each resulting cell intersects at most one line in $\{x=6i\delta_0:i\in \mathbb{Z}\}$
and at most one line in $\{y=6i\delta_0:i\in \mathbb{Z}\}$. Consequently, the $x$- and $y$-projection of each resulting cell is an interval of length at most $12\delta_0$. Hence the diameter of the final cells is at most $12\sqrt{2}\delta_0\leq 12\sqrt{2}(1+\epsilon) \d^*$.
\end{proof}

\section{$O(\log\opt)$-approximation for \mlc}\label{sec:circle}
This section considers the radius of the largest inscribed circle as the measure of cell size; in particular, in \mlc the goal is to split the polygon \P (which may have holes) into pieces so that no piece contains a disk of radius 1. We give an $O(\log\opt)$-approximation algorithm for \mlc based on reducing the problem to SetCover.
The following reformulation is crucial for the approximation algorithm:
\begin{observation}\label{obs:reformulation}
A set of lasers splits \P into pieces of in-circle radius at most 1 iff every unit disk that lies inside \P is hit by a laser.
\end{observation}

\begin{theorem}\label{thm:apxcircle}
For a polygon \P with $n$ vertices (possibly with holes),
\mlc admits an $O(\log\opt)$-approximation in time polynomial in $n$ and $\area(P)$.
\end{theorem}
\begin{proof}
We lay out a regular square grid of points at spacing of $\sqrt{2}$. The set of grid points within $P$ is denoted by $G$. We may assume $|G|=O(\area(P))$ by a suitable (e.g., uniformly random) shift.
%They are spaced so that no unit-radius disk can avoid containing a point of $G$ (possibly on its boundary).
Due to the spacing, every unit-radius disk in $P$ contains a point of $G$ (possibly on its boundary).

Consider an optimal set, $L^*$ of lasers that hit all unit disks that are contained within $P$.
Replace each laser (chord) $c\in L^*$ with up to four anchored chords of the same homotopy type as $c$
with respect to the vertices of $P$ and the points $G$, obtained as follows:
%-- the anchored chords are obtained by rotating $c$ clockwise/counterclockwise, while keeping its endpoints on the same pair of edges of $P$, and not letting $c$ pass over a vertex of $P$ or a grid point of $G$.
%\rathishESA{rotating around the center point of chord $c$?}
Shift the chord $c$ vertically down (up), while keeping its endpoints on the same pair of edges of $P$, until it becomes incident to a point in $G$ or a vertex of $P$, then rotate the chord clockwise (counterclockwise) around this point until it becomes incident to another point in $G$ or a vertex of $P$.
Since every unit disk within \P contains a point of $G$, any unit disk within $P$ that is intersected by $c$ is also intersected by one of the shifted and rotated copies of $c$.
This means that we can construct a candidate set, $C$, of $O((n+\area(P))^2)$ chords that can serve as lasers in an  approximate solution, giving up at most a factor 4 of optimal.
Further, in the arrangement of the segments $C$ within $P$, any unit disk is intersected by some set of chords of $C$, thereby defining a combinatorial type for each unit disk in $P$. (Two disks are of the same type if they are intersected by the same subset of chords in $C$; one way to define the type is to associate it with a cell in the arrangement of lines drawn parallel to each chord $c\in C$ at distance 2 from $c$ on each side of $c$ -- while the center of the disk is in one cell of the arrangement, the disk intersects the same chords.) Let $D$ be the polynomial-size ($O(|C|^2)$) set of disks, one ``pinned'' (by two segments, from the set $C$ and the set of edges of $P$) disk per combinatorial type.
By construction, any set of chords from $C$ that meets all disks of $D$ must meet all unit disks within $P$.

We thus formulate a discrete set cover instance in which the ``sets'' correspond to the candidate set $C$ of chords, and the ``elements'' being covered are the disks $D$.
Since there are constant-size sets of disks that cannot be shattered, the VC dimension of the set system is constant, and an $O(\log\opt)$-approximate solution for the set cover can be found in time polynomial in the size of the instance \cite{BG}.
\end{proof}
The same algorithm works for the version in which the lasers are restricted to be axis-aligned (the only change is that the candidate set consists from the grid of axis-aligned chords through $G$ and vertices of \P).

% !TEX root =  main.tex

\paragraph{Acknowledgements.}
We thank Peter Brass for technical discussions and for organizing an NSF-funded workshop where these problems were discussed and this collaboration began.
This research was partially supported by NSF grants
CCF-1725543, %Support for Rathish (Michael Bender's grant)
CSR-1763680,  %Support for Rathish (Michael Bender's grant)
CCF-1716252,  %Support for Rathish (Michael Bender's grant)
CCF-1617618, 
CCF-1439084,
CNS-1938709,   
DMS-1800734, % Csaba's grant
CRII-1755791, % Mayanks's grant
CCF-1910873,  % Mayanks's grant
CNS-1618391, % Jie's grant
DMS-1737812, % Jie's grant
OAC-1939459, % Jie's grant
CCF-1540890, and % Estie and Joe's grant
CCF-2007275. % Estie and Joe's grant
The authors also acknowledge partial support from the US-Israel Binational Science Foundation (project 2016116), the DARPA Lagrange program, the Sandia National Labs and grants by the Swedish Transport Administration and the Swedish Research Council.
%$\cdots \cdots$ 
\bibliographystyle{abbrv}
\bibliography{laser}

\begin{thebibliography}{10}

\bibitem{armaselu2015algorithms}
B.~Armaselu and O.~Daescu.
\newblock Algorithms for fair partitioning of convex polygons.
\newblock {\em Theoretical Computer Science}, 607:351--362, 2015.

\bibitem{barany2010equipartitioning}
I.~B{\'a}r{\'a}ny, P.~Blagojevi{\'c}, and A.~Sz{\H{u}}cs.
\newblock Equipartitioning by a convex 3-fan.
\newblock {\em Advances in Mathematics}, 223(2):579--593, 2010.

\bibitem{Bar02}
A.~I. Barvinok.
\newblock {\em A course in convexity}, volume~54 of {\em Graduate Studies in
  Mathematics}.
\newblock AMS, Providence, RI, 2002.

\bibitem{bezdek1995solution}
A.~Bezdek and K.~Bezdek.
\newblock A solution of {C}onway's fried potato problem.
\newblock {\em Bulletin of the London Mathematical Society}, 27(5):492--496,
  1995.

\bibitem{Blagojevic2014}
P.~V.~M. Blagojevi{\'{c}} and G.~M. Ziegler.
\newblock Convex equipartitions via equivariant obstruction theory.
\newblock {\em Israel Journal of Mathematics}, 200(1):49--77, 2014.

\bibitem{Borsuk32}
K.~Borsuk.
\newblock Drei {S}\"atze \"uber die $n$-dimensionale euklidische {S}ph\"are.
\newblock {\em Fundamenta Mathematicae}, 20:177--190, 1933.

\bibitem{bose1998polygon}
P.~Bose, J.~Czyzowicz, E.~Kranakis, D.~Krizanc, and A.~Maheshwari.
\newblock Polygon cutting: Revisited.
\newblock In {\em Proc. Japanese Conference on Discrete and Computational
  Geometry}, volume 1763 of {\em LNCS}, pages 81--92. Springer, 1998.

\bibitem{BG}
H.~Br{\"o}nnimann and M.~T. Goodrich.
\newblock Almost optimal set covers in finite {VC}-dimension.
\newblock {\em Discrete and Computational Geometry}, 14(4):463--479, 1995.

\bibitem{chazelle1982theorem}
B.~Chazelle.
\newblock A theorem on polygon cutting with applications.
\newblock In {\em Proc. 23rd IEEE Symposium on Foundations of Computer Science
  ({FOCS})}, pages 339--349, 1982.

\bibitem{croftunsolved}
H.~T. Croft, K.~J. Falconer, and R.~K. Guy.
\newblock {\em Unsolved Problems in Geometry}.
\newblock Springer-Verlag, New York, 1991.

\bibitem{das1991geometric}
G.~Das and G.~Narasimhan.
\newblock Geometric searching and link distance.
\newblock In {\em Proc. 2nd Workshop on Algorithms and Data Structures
  ({WADS})}, volume 519 of {\em LNCS}, pages 261--272. Springer, 1991.

\bibitem{FreemanS75}
H.~Freeman and R.~Shapira.
\newblock Determining the minimum-area encasing rectangle for an arbitrary
  closed curve.
\newblock {\em Commun. {ACM}}, 18(7):409--413, 1975.

\bibitem{Freimer-et-al}
R.~Freimer, J.~S.~B. Mitchell, and C.~Piatko.
\newblock On the complexity of shattering using arrangements.
\newblock Technical report, Cornell University, 1991.

\bibitem{Gopinathan2003-oo}
U.~Gopinathan, D.~J. Brady, and N.~Pitsianis.
\newblock Coded apertures for efficient pyroelectric motion tracking.
\newblock {\em Opt. Express}, 11(18):2142--2152, 2003.

\bibitem{GH2005}
R.~Gu{\`a}rdia and F.~Hurtado.
\newblock On the equipartition of plane convex bodies and convex polygons.
\newblock {\em Journal of Geometry}, 83(1):32--45, 2005.

\bibitem{Gustafson1982-ng}
S.~C. Gustafson.
\newblock Intensity correlation techniques for passive optical device
  detection, 1982.

\bibitem{HassinM91}
R.~Hassin and N.~Megiddo.
\newblock Approximation algorithms for hitting objects with straight lines.
\newblock {\em Discrete Applied Mathematics}, 30(1):29--42, 1991.

\bibitem{He2004-xs}
T.~He, Q.~Cao, L.~Luo, T.~Yan, L.~Gu, J.~Stankovic, and T.~Abdelzaher.
\newblock Electronic tripwires for power-efficient surveillance and target
  classification.
\newblock In {\em Proc. 2nd International Conference on Embedded Networked
  Sensor Systems ({SenSys} 2004)}. ACM Press, 2004.

\bibitem{suri}
J.~Hershberger and S.~Suri.
\newblock A pedestrian approach to ray shooting: Shoot a ray, take a walk.
\newblock {\em J. Algorithms}, 18(3):403--431, 1995.

\bibitem{JenrichB14}
T.~Jenrich and A.~E. Brouwer.
\newblock A 64-dimensional counterexample to {B}orsuk's {C}onjecture.
\newblock {\em Electr. J. Comb.}, 21(4):P4.29, 2014.

\bibitem{John48}
F.~John.
\newblock Extremum problems with inequalities as subsidiary conditions.
\newblock In {\em Studies and Essays Presented to R.~Courant on his 60th
  Birthday}, pages 187--204. Interscience Publishers, Inc., New York, NY, 1948.

\bibitem{KK93}
J.~Kahn and G.~Kalai.
\newblock A counterexample to {Borsuk's} conjecture.
\newblock {\em Bull. Amer. Math. Soc.}, 29:60--62, 1993.

\bibitem{Karasev2014}
R.~Karasev, A.~Hubard, and B.~Aronov.
\newblock Convex equipartitions: the spicy chicken theorem.
\newblock {\em Geometriae Dedicata}, 170(1):263--279, 2014.

\bibitem{Keil00}
J.~M. Keil.
\newblock Polygon decomposition.
\newblock In J.~Sack and J.~Urrutia, editors, {\em Handbook of Computational
  Geometry}, pages 491--518. North Holland/Elsevier, 2000.

\bibitem{kostitsyna2015balanced}
I.~Kostitsyna.
\newblock {\em Balanced partitioning of polygonal domains}.
\newblock PhD thesis, Stony Brook University, Stony Brook, NY, 2015.

\bibitem{LangermanM05}
S.~Langerman and P.~Morin.
\newblock Covering things with things.
\newblock {\em Discrete {\&} Computational Geometry}, 33(4):717--729, 2005.

\bibitem{LinkDistance}
A.~Maheshwari, J.-R. Sack, and H.~N. Djidjev.
\newblock Link distance problems.
\newblock In J.-R. Sack and J.~Urrutia, editors, {\em Handbook of Computational
  Geometry}, chapter~12, pages 519--558. North-Holland, 2000.

\bibitem{MegiddoT82}
N.~Megiddo and A.~Tamir.
\newblock On the complexity of locating linear facilities in the plane.
\newblock {\em Operations Research Letters}, 1:194--197, 1982.

\bibitem{mitchell2014minimum}
J.~S. Mitchell, V.~Polishchuk, and M.~Sysikaski.
\newblock Minimum-link paths revisited.
\newblock {\em Computational Geometry}, 47(6):651--667, 2014.

\bibitem{nandakumar2012fair}
R.~Nandakumar and N.~{Ramana~Rao}.
\newblock Fair partitions of polygons: An elementary introduction.
\newblock {\em Proceedings-Mathematical Sciences}, 122(3):459--467, 2012.

\bibitem{robertson1998cake}
J.~Robertson and W.~Webb.
\newblock {\em Cake-cutting algorithms: Be fair if you can}.
\newblock AK Peters/CRC Press, 1998.

\bibitem{soberon2012}
P.~Sober\'on.
\newblock Balanced convex partitions of measures in $\mathbb{R}^d$.
\newblock {\em Mathematika}, 58(1):71--76, 2012.

\bibitem{Suri90}
S.~Suri.
\newblock On some link distance problems in a simple polygon.
\newblock {\em {IEEE} Trans. Robotics and Automation}, 6(1):108--113, 1990.

\bibitem{Toussaint14}
G.~Toussaint.
\newblock Applications of the rotating calipers to geometric problems in two
  and three dimensions.
\newblock {\em International Journal of Digital Information and Wireless
  Communications}, 4(3):372--386, 2014.

\bibitem{Zahnd_undated-pw}
S.~{Zahnd}, P.~{Lichisteiner}, and T.~{Delbruck}.
\newblock Integrated vision sensor for detecting boundary crossings.
\newblock In {\em 2003 IEEE International Symposium on Circuits and Systems
  ({ISCAS})}, volume~2, 2003.

\bibitem{Zheng2007-fl}
Y.~Zheng, D.~J. Brady, and P.~K. Agarwal.
\newblock Localization using boundary sensors: An analysis based on graph
  theory.
\newblock {\em ACM Trans. Sen. Netw.}, 3(4), 2007.

\end{thebibliography}
%\clearpage
\appendix

\end{document}